\documentclass[12pt]{report}

\usepackage{amsthm, amsmath, amssymb, amsfonts, algorithm, cite, multicol,stddef,upgreek}
\usepackage{tikz,graphicx,subcaption}
\usepackage{tikz-cd}
\usetikzlibrary{decorations.pathreplacing}
\usepackage{cite}
\usepackage[capitalize]{cleveref}
\usepackage{amsmath,amssymb,amsfonts,thm-restate}
\usepackage{algorithmic}
\usepackage{graphicx}
\usepackage{textcomp}
\usepackage{xcolor}
\def\BibTeX{{\rm B\kern-.05em{\sc i\kern-.025em b}\kern-.08em
    T\kern-.1667em\lower.7ex\hbox{E}\kern-.125emX}}
\usepackage{mythesis,uwthesis}

\numberwithin{equation}{section}
\numberwithin{theorem}{section}
\numberwithin{subsubsection}{subsection}
\numberwithin{figure}{section}

\begin{document}
\title{EFFICIENTLY ESTIMATING A SPARSE \protect\\
      DELAY-DOPPLER CHANNEL.}
\author{Alisha Zachariah}
\oraldate{May 29, 2020}
\profA{S. Gurevich, Associate Professor, Mathematics}
\profB{N. Boston, Professor Emeritus, Mathematics and ECE}
\profC{B. Lesieutre, Professor, Electrical and Computer Engineering (ECE)}
\profD{S. Goldstein, Researcher, Botany }
\degree{Doctor of Philosophy}
\dept{Mathematics}
\thesistype{dissertation}
\beforepreface
\prefacesection{Abstract}
Multiple wireless sensing tasks, e.g., radar detection for driver safety, involve estimating the ``channel” or relationship between signal transmitted and received. In this paper, we focus on a certain channel model known as the \textit{delay-doppler channel}. This model begins to be useful in the \textit{high frequency carrier} setting, which is increasingly common with developments in millimeter-wave technology. Moreover, the delay-doppler model then continues to be applicable even when using signals of \textit{large bandwidth}, which is a standard approach to achieving \textit{high resolution} channel estimation. However, when high resolution is desirable, this standard approach results in a tension with the desire for efficiency because, in particular, it immediately implies that the signals in play live in a space of very high dimension N (e.g., $\sim10^6$ in some applications), as per the Shannon-Nyquist sampling theorem.

To address this difficulty, in this paper we propose a novel randomized estimation scheme called \textit{Sparse Channel Estimation}, or \textit{SCE} for short, for channel estimation in the $k$-sparse setting (e.g., $k$ objects in radar detection). This scheme involves an estimation procedure with sampling and space complexity both on the order of $k(\log N)^3$, and arithmetic complexity on the order of $k(\log N)^3+k^2$, for $N$ sufficiently large.

To the best of our knowledge, Sparse Channel Estimation (SCE) is the first of its kind to achieve these complexities simultaneously -- it seems to be extremely efficient! As an added advantage, it is a simple combination of three ingredients, two of which are well-known and widely used, namely digital chirp signals and discrete Gaussian filter functions, and the third being recent developments in sparse fast fourier transform algorithms.

We note that the design of Sparse Channel Estimation(SCE) is based on a digital channel model which presumes the so-called \textit{``on-the-grid"} assumption. This assumption is made in a heuristic manner -- while SCE can still be used in the \textit{off-the-grid} setting, proving guarantees in this more general setting remains to be future work.
 \prefacesection{Acknowledgements}
 
I would like to thank my thesis advisor, Shamgar Gurevich, for his advice and guidance, and the enormous amount of time he has invested in my growth as a scientist. 

I am extremely grateful to Nigel Boston and Bernie Lesieutre who have also been advisors to me and supported my work in a multitude of ways, over the course of our project on algebraic methods for the power flow equations. 

I would also like to thank Steven Goldstein and Alexander Fish for crucial discussions that informed this thesis project.

Finally, I would like to thank the Mathematics Department -- my personal experience here has been one of receiving tremendous support for the pursuit of my research. 

\listoffigures
\listoftables
%

\afterpreface
%
\chapter{Introduction}\label{introduction}   
%
%
%
%
%
%

The process of channel estimation is going on all around us. 
For instance, all our personal devices are constantly sending and receiving signals and, intuitively speaking, the ``channel'' is simply the relationship between the signal transmitted and signal received, certain parameters of which we would like to estimate.

In this paper we focus on the \textit{delay-doppler channel}. Within that setting, our methods are quite general and may serve a variety of wireless communication applications, but to illustrate the key ideas, the rest of this work will deal with the application to radar detection. 

We begin with an intuitive physical picture of the radar task. \cref{sec:mot,sec:cont,sec:dig} will closely follow standard references such as \cite{levanon2004radar,tse2005fundamentals}.

\section{\textbf{Motivation and Intuitive Physical Picture}}\label{sec:mot}

Consider the classical problem of estimating the position and velocity of some object of interest.

\begin{figure}[h]
\centerline{\includegraphics[width=0.5\textwidth]{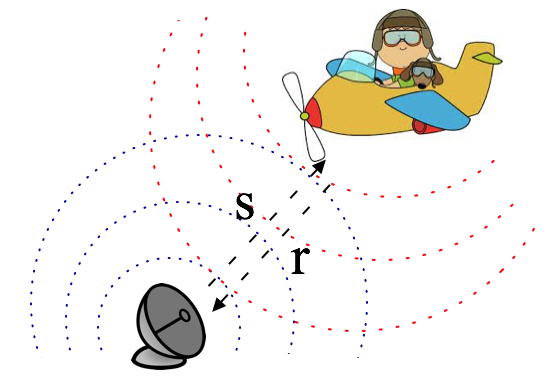}}
\caption{Radar detection of object in the case of line of sight.}
\label{fig1}
\end{figure}

In practice, we could approach this problem using a device called a radar
which emits electromagnetic waves in all directions around it -- see \cref{fig1}. The emitted waves or ``signal" travels through space and will be reflected back to the radar by certain materials, such as that of our object. 

For the rest of this paper, we assume a 
\textit{line of sight} between the radar and object -- as illustrated in Fig.\ref{fig1}. The reflection from the object back to the radar is then strongest along this direction. Restricting
to the line of sight, we denote the signal \textit{transmitted} by $s$ and the signal reflected or \textit{received} by $r$. 

The signals $s$ and $r$ are physically related, and this can be utilized to estimate position and velocity of the object of interest. More precisely, a ``digital" computational procedure will yield the parameters of interest.

We will eventually describe certain computational challenges in the digital estimation step. However, to see exactly where these challenges come from, it will be to our advantage to first consider a \textit{continuous} (or \textit{analog}) model of the  channel relationship just intuitively described, together with the corresponding estimation task.

\section{\textbf{Continuous Channel Model and Estimation Task}}\label{sec:cont}

Engineers tell us that the transmitted and received signals,  $s$ and $r$, depend on \textit{time} and have finite \textit{energy}. As a consequence, it is natural to model them as elements of the Hilbert space $L^2(\R)$, of complex-valued functions of one real variable with bounded $L^2$-norm with respect to the standard inner product $\<\cdot, \cdot\>$ \cite{levanon2004radar, tse2005fundamentals}. For the reader's benefit, we may denote an element $s\in L^2(\R)$ by $s(t)$ where $t\in \R$ models time. We refer to elements of this space as \textit{continuous} (or \textit{analog}) \textit{signals}. 

Next, we would like to model the relationship between $s$ and $r$. \\
Firstly, it takes time for the signal transmitted to travel the distance to the object and back. This is modeled as a time delay, i.e., 
\begin{align*}
    r(t) & = \alpha_0 \cdot s(t-t_0) + \text{Noise},
\end{align*}
More precisely, let us denote the distance (often called \textit{relative range}) between the radar and object by $d_0$ -- see Fig. \ref{fig2}, then $t_0 = \frac{2d_0}{c}$ with $c$ denoting the speed of light. Moreover, $\alpha_0$ is a complex-valued scalar with magnitude $|\alpha_0|\le 1$, known as the \textit{attenuation coefficient}, and models loss of energy. Finally, the additive noise term that appears in the equation is intended to account for imprecisions in the model, including environmental effects. 

Next -- and most interestingly -- when the object is moving, the signal is also subject to the \textit{Doppler effect}. \\ 
In general, the Doppler effect can be modeled as a \textit{time scale}, as detailed in \cref{app:timescale}. However, we will adopt the standard \textit{narrowband} assumption \cite{levanon2004radar,tse2005fundamentals}. Namely, denoting the bandwidth (i.e., size of support of the Fourier transform) of $s$ by $W$ and its carrier (i.e., central) frequency by $f_c$, we assume that $W \ll f_c$ --  for a precise quantification of what  ``$\ll$" can mean in this context, see \cref{app:tstofs}. \\
Under this assumption, the Doppler effect can be modeled by a frequency shift \cite{levanon2004radar}, so that we have
\begin{align}\label{eq1}
    r(t) & = \alpha_0 \cdot e^{2\pi if_0t}s(t-t_0) + \text{Noise}, 
\end{align}
where, $f_0 = -\frac{2f_cv_0}{c}$ with $v_0$ denoting the object's \textit{relative radial velocity}, as illustrated in Fig. \ref{fig2}. A mathematical justification for approximating a time scale by a frequency shift, under the narrowband assumption, can be found in \cref{app:tstofs}. \\
\begin{figure}[H]
  \centering    
     \includegraphics[width=0.8 \textwidth]{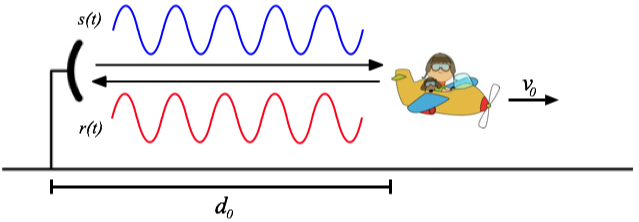}
     \caption{Illustrating range and radial velocity of object.}
     \label{fig2}
\end{figure} 
\ Note that if $f_c$ is not sufficiently large, the quantity $f_0$ is negligible and often ignored. However, with a high carrier frequency (i.e., in the millimeter-wave setting) \cite{niu2015survey}, $f_0$ may be on the order of \textit{megahertz} (MHz) 
-- even for relatively slow moving objects like cars.\\ So, to summarize, given a good estimate of $(t_0,f_0)$, we will be able to estimate range and radial velocity $(d_0,v_0)$.


However, in general we may have more than one object of interest for which we would like to estimate range and radial velocity. In this case, the received signal is a \textit{superposition} of reflections from each object \cite{tse2005fundamentals}, i.e., the relationship between $r$ and $s$ is given by 
\begin{align}\label{eq:cont}
    r(t) & = \sum_{j=1}^k \alpha_j \cdot h_{t_j,f_j}s(t) + \text{Noise}, 
\end{align}
where, $h_{t_j,f_j}$ for $(t_j,f_j)\in \R^2$, denotes the \textit{continuous time-frequency shift} operator on $L^2(\R)$ given by,
\begin{align*}
    h_{t_j,f_j}s(t) &= e^{2\pi if_jt}s(t-t_j).
\end{align*}
and the parameter $k$, known as \textit{channel sparsity}, models the number of \textit{targets} in the case of radar detection. \\
We will refer to (\ref{eq:cont}) as the \textit{continuous channel model}, and we can now formulate the following estimation task.
\begin{task}[\textbf{\textit{Continuous Estimation}}]\label{CET}
Assuming \cref{eq:cont}, estimate $(t_j,f_j)$ for $j=1, \dots , k$.
\end{task}
We have the following remarks regarding \cref{CET}.
\begin{remark}
Our motivating problem was to estimate position and velocity, rather than simply range and radial velocity, as with \cref{CET}. However, in theory, it is possible to estimate position and velocity by performing \ref{CET} with four non-coplanar radars and solving a small system of quadratic equations (four equations in three variables), for $k$ targets in generic position \cite{mailloux2017phased}. In practice, for robustness, this procedure may be generalized to using an array of radars (for instance, a phased array radar \cite{mailloux2017phased}).
\end{remark}
\begin{remark}
While \cref{CET} assumes that the channel sparsity $k$ is known, in practice this may not true. For instance we may only know some upper bound on $k$. When $k$ is unknown, we may rather wish to estimate shifts $(t,f)\in \R^2$, whose coefficients $\alpha$ are "significant" under suitable assumptions on noise. The algorithm we propose is designed to perform this modified estimation task, but for ease of exposition, we work with \cref{CET} for now. 
\end{remark}

We are now ready to move towards the promised digital computational procedure.

\section{\textbf{Digital Channel Model and Estimation Task}}\label{sec:dig}

In order to produce a computational procedure, we wish to reduce the continuous estimation task to a finite dimensional linear algebra problem, and \textit{digital signal processing} (DSP) allows us to do just that \cite{tse2005fundamentals}.\\
In DSP -- see \cref{fig:sq} for illustration -- the continuous signal $s$ will begin life as a \textit{digital signal} $S$ and, similarly, the continuous signal $r$ ends up as  digital signal $R$. So we first provide a standard description for these.

\vspace{.25cm}
\subsection{The space of digital signals}
Fix a positive integer $N$ and consider the set $\Set{0,1,\dots, N-1}$ that we denote by $\Z_N$; recall that it comes naturally equipped with addition and multiplication modulo $N$. \\
Now, just as continuous signals were elements of $L^2(\R)$, digital signals will be elements of the space $L^2(\Z_N)$, of complex-valued functions\footnote{Equivalently, the space of $N$-periodic complex-valued functions on the integers $\Z$.} on $\Z_N$, with its natural inner product\footnote{ Namely, $\<S_1,S_2\> = \sum_{\tau\in \Z_N} S_1[\tau] \cdot \overline{S_2[\tau]}$ for every $S_1,S_2$ in $L^2(\Z_N)$.} $\<\cdot, \cdot \>$.

Having a description for digital signals at our disposal, we can now describe what it means for a continuous signal to ``begin life" or ``end up" as a digital one, and fully  flesh out the stages pictured in \cref{fig:sq}.

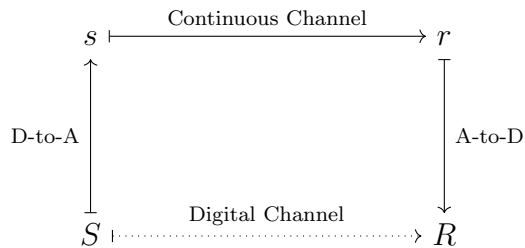
\begin{figure}[H]
    \centering
    \begin{tikzpicture}[descr/.style={fill=white}]
    \matrix(m)[matrix of math nodes, row sep=1em, column sep=2.5em,
    text height=1.5ex, text depth=0.25ex]
    {& s &&&& r \\ &&&&&\\ &&&&&\\ &&&&&\\ &&&&& \\ & S &&&& R \\};
    7
    \path[|->,font=\scriptsize]
    (m-1-2) edge node[above] {Continuous Channel} (m-1-6)
    (m-6-2) edge node[left] {D-to-A} (m-1-2);
    \path[|->,font=\scriptsize]
    (m-1-6) edge node[right] {A-to-D} (m-6-6);
    \path[|->,dotted, font=\scriptsize]
    (m-6-2) edge node[above] {Digital Channel}(m-6-6);
    \end{tikzpicture}
    \caption{Signal life-cycle in DSP.}
    \vspace*{-0.07in}
    \label{fig:sq}
\end{figure}

\subsection{Moving between digital and analog settings}
The process by which $s\in L^2(\R)$ is generated from $S\in L^2(\Z_N)$, is called \textit{digital-to-analog} ($\text{D-to-A}$), 
as illustrated on the left-hand side of \cref{fig:sq}. Similarly, the process by which the received signal $r$ ends up as a digital signal $R$, is called \textit{analog-to-digital}, illustrated on the right-hand side of \cref{fig:sq} 
($\text{A-to-D}$). \\ There are various options for these processes \cite{oppenheim1983signals} and in this work we adopt a choice that is frequently used in the literature  \cite{tse2005fundamentals}. This approach is typically attributed to Shannon, and was introduced in his seminal work \cite{shannon1949communication}. We proceed to provide a brief description of Shannon's approach below.

Let's first recall that the signal transmitted $s$, has bandwidth $W$. In addition, let's denote the \textit{duration}\footnote{As a consequence of the \textit{uncertainty principle}, signals cannot simultaneously be \textit{time-limited} and \textit{bandlimited}. In practice, signals will be \textit{essentially} time-limited and bandlimited.} of $s$ by $T$. Shannon showed that the space of continuous signals with this duration and bandwidth, is close to being $N$ dimensional for $N=TW$ (for simplicity, we assume $TW$ is an integer), in a sense which we do not make explicit here. Along the way, he provided an explicit formula for a linear map,
\begin{align}\label{eq:da}
    \text{D-to-A}&: L^2(\Z_N) \to L^2(\R), \\
    S&\mapsto s,
\end{align}
that would produce signals of duration $T$ and bandwidth $W$. In addition, he provided exact formulas for a linear map,
\begin{align}\label{eq:ad}
   \text{A-to-D}&: L^2(\R) \to L^2(\Z_N), \\
   r &\mapsto R
\end{align}
\
(In fact these maps model certain physical procedures, which have been implemented in devices in order to generate and process continuous signals \cite{oppenheim1983signals}.)

While we do not state the formulas that realize maps \ref{eq:da} and \ref{eq:ad} at this point, we next present a key property of these maps that will enable us to perform the continuous estimation task \ref{CET} \textit{digitally}.  For the reader's reference, exact formulas for $\text{D-to-A}$ and $\text{A-to-D}$ can be found in \cref{app:digital-analog}, and their key attribute can be immediately derived from these. \\
Note that, in describing this property, we will refer to a certain \textit{grid} in $\R^2$, namely, $\frac{1}{W}\Z\times \frac{1}{T}\Z$, as the \textit{time-frequency grid}, denoted
\begin{align}
    \Lambda_{T,W} = \frac{1}{W}\Z\times \frac{1}{T}\Z.
\end{align}
In addition, $H_{\tau_0,\omega_0}$, for $(\tau_0,\omega_0)\in \Z^2_N$, will denote the \textit{discrete time-frequency shift} operator on $L^2(\Z_N)$ given by,
\begin{align*}
    H_{\tau_0,\omega_0}S[\tau] & = e^{2\pi i\frac{\omega_0\tau}{N}}S[\tau-\tau_0]. 
\end{align*}
We are now ready to state the property of interest.
\begin{restatable}{prop}{propadda}\label{prop:ad-da}
For a continuous shift on the time-frequency grid, $(t_0,f_0) \in \Lambda_{T,W}$, there exists $(\tau_0,\omega_0)$ in $\Z_N^2$ such that 
\begin{align}
    \text{A-to-D}\circ h_{t_0,f_0}\circ \text{D-to-A} &= H_{\tau_0,\omega_0}
\end{align}
\end{restatable}
We have the following remark regarding Property \ref{prop:ad-da}.
\begin{remark}[Boundedness assumption]\label{rem:recovery}If in addition to lying on the grid $\Lambda_{T,W}$, we assume, for instance, that $t_0\in [0,T]$ and $f_0\in[-W/2,W/2]$, then the continuous shift $(t_0,f_0)$ can be uniquely recovered from the discrete shift $(\tau_0,\omega_0)$ with which it is associated by Property \ref{prop:ad-da}. For simplicity of exposition, we will adopt the boundedness assumption for the rest of this work.
\end{remark}

With \ref{prop:ad-da} in hand, we can now move towards recasting the continuous estimation task \ref{CET} as a digital one.

\vspace{.25cm}
\subsection{Digital model and estimation task}
The first step towards a digital estimation task is formulating a model for the digital channel relationship between the ``transmitted" and ``received" signals, $S$ and $R$. 

The relationship between $S$ and $R$ is induced -- as we see from following the arrows of \cref{fig:sq} -- by the continuous channel relationship \ref{eq:cont} between $s$ and $r$. If we assume that $(t_j,f_j)$ in \cref{eq:cont} lie on the grid $\Lambda_{T,W}$ then, by \cref{prop:ad-da}, the relationship between $R$ and $S$ can be modeled as,
\begin{align}\label{eq:digital}
    R[\tau] & = \sum_{j=1}^k \alpha_j \cdot H_{\tau_j,\omega_j}S[\tau] + \text{Noise}
\end{align}
for $(\tau_j,\omega_j) \in \Z_N^2$. Moreover, under suitable assumptions (see \cref{rem:recovery}), $(t_j,f_j)$ can be uniquely recovered from $(\tau_j,\omega_j)$. 

\begin{remark}[On-the-grid assumption]
In practice, the continuous shifts $(t_j,f_j)$ in the channel model \ref{eq:cont}, will never lie exactly on the time-frequency grid $\Lambda_{T,W}$. However, we make the ``on-the-grid" assumption in a heuristic manner -- any resulting estimation scheme could still be applied in the general setting and, if one expects some form of continuity in the model, a continuous shift $(t_j,f_j)$ could be identified by the point on the grid closest to it. Note that we will adopt the on-the-grid assumption for the rest of this work.
\end{remark}\
We will refer to \ref{eq:digital} as the \textit{digital channel model}.

At first glance, the digital model \ref{eq:digital} might seem unsatisfactory since it only \textit{resolves} or distinguishes two shifts $(t_1,f_1)$ and $(t_2,f_2)$ that are at least $\p{\frac{1}{W},\frac{1}{T}}$ apart but, if the resolution $\p{\frac{1}{W},\frac{1}{T}}$ is sufficiently small -- equivalently, if $N=TW$ is sufficiently large -- 
this model can begin to be useful. 
For instance, state-of-the-art radars aim to resolve objects at a distance of as little as \textit{centimeters} apart \cite{schneider2005automotive}. So, in such applications, the model starts to be useful for bandwidth on the order of $10^{9} Hz$ or, equivalently, $N$ on the order of $10^6$ for signals of duration in \textit{milliseconds}.

We make a remark regarding this standard model.
\begin{remark}
A characteristic of \ref{eq:digital}, is that if we assume this model, our ability to resolve time-frequency shifts is limited by the bandwidth and duration of the signal transmitted. However, since \ref{eq:digital} is a consequence of our choice of D-to-A and A-to-D, it is unclear to the author whether this limitation is fundamental in nature, or simply a consequence of that choice.
\end{remark} 
In summary of the above discussion, we can reduce the continuous estimation task \ref{CET} to the following:
\begin{task}[\textbf{\textit{Digital Estimation}}]\label{DET}
Assuming \cref{eq:digital}, detect $(\tau_j,\omega_j)$ for $j=1, \dots , k$.
\end{task}
We need an \textit{estimation scheme} in order to perform the task \ref{DET}, and in this work, we will assume certain ``Rules of the Game" for any such scheme. 

For the reader's benefit, we introduce some notation that will be helpful in describing these rules. \\
The evaluation of a signal $S\in L^2(\Z_N)$ at $\tau\in \Z_N$ will be referred to as \textit{sampling} $S$, and the value $S[\tau]$ will be called a \textit{sample} of $S$. In addition, we will denote the \textit{channel operator} by $H$, namely,
\begin{align*}
    H &= \sum_{i=1}^k \alpha_i \cdot H_{\tau_i, \omega_i}.
\end{align*}
\subsection{Rules of the digital estimation game} 
In this work, we only study/design digital estimation schemes which may involve:
\begin{itemize}
       \item A choice of ``mechanism" $\mathbb{P}$ for picking an element in $L^2(\Z_N)$. More precisely, we will be able to choose a \textit{probability distribution} $\mathbb{P}$ on $L^2(\Z_N)$.
    \item A choice of estimation algorithm which may involve: 
    \begin{itemize}
        \item Picking a digital signal $S\in L^2(\Z_N)$ using $\mathbb{P}$.
        \item Storing samples of elements of $L^2(\Z_N)$ as needed, such as $S$ and
        \begin{align*}
            R = HS +\text{Noise},
        \end{align*}
        \item Correlating and/or applying linear operators to $S$ and $R$.
    \end{itemize}
\end{itemize}
\vspace{.25cm}
In the next chapter, we describe existing schemes for the given digital estimation game.
%
\chapter{Existing Methods and Computational Challenges}   
%
%
%
%
%
%

We now elaborate existing estimating schemes to perform the digital estimation task (\ref{DET}).

\section{\textbf{Matched Filter and Pseudorandom Method}}\label{existing:prm}
Given the ``Rules of the Game" in \cref{sec:dig}, we might choose an estimation scheme which involves picking a signal $S$ such that, for $R=H_{\tau_0,\omega_0}S$, the inner products  $\<R,H_{\tau,\omega}S\>$ have a distinguished maximum (in magnitude) at $(\tau,\omega) = (\tau_0,\omega_0)$. 

For convenience, we introduce some standard notation \cite{levanon2004radar} at this stage.\\ Given any two signals in $L^2(\Z_N)$ we can define their \textit{ambiguity function} as follows.
\begin{definition} We define the \underline{ambiguity function} of $S$ against $R$, denoted $\mathcal{A}(S,R)$, on the $\Z_N^2$ plane as follows
\begin{align*}
    \mathcal{A}(S,R)[\tau,\omega] &= \<H_{\tau,\omega}S,R\> & \tau,\omega \in \Z_N^2.
\end{align*}
We denote the ambiguity function of a signal $S$ against itself, simply as $\mathcal{A}(S)$.
\end{definition}
Such functions are sometimes referred to as a \textit{matched filter}. \\ Now that we have this notation in hand, note that $\mathcal{A}$ satisfies the following convenient property.
\begin{prop}[Shifts]\label{prop:amb2}
It follows from the fact that operators $H_{\tau,\omega}$ are unitary that,
\begin{align*}
    \mathcal{A}\p{S,H_{\tau_0,\omega_0}S}=\mathcal{A}(S)[\tau-\tau_0,\omega-\omega_0]
\end{align*}
for $S,R\in L^2(\Z_N)$ and $(\tau,\omega), (\tau_0,\omega_0) \in \Z_N^2$.
\end{prop}
Now recall our original premise -- we would like to pick a signal $S$ such that, for $R=H_{\tau_0,\omega_0}S$, the inner products  $\<R,H_{\tau,\omega}S\>$ have a distinguished peak at $(\tau,\omega) = (\tau_0,\omega_0)$. By \cref{prop:amb2}, restating this premise using the notation of the ambiguity function, we would like to pick a signal $S$ such that $\mathcal{A}(S)$ has a distinguished peak at $(\tau,\omega) = (0,0)$. Interestingly, such signals are readily available -- if we pick a signal from $L^2(\Z_N)$ \textit{at random}, in a sense that we do not make precise here, it will satisfy this requirement almost surely. In practice, a \textit{pseudorandom} signal is used \cite{fish2013delay}, and so the resulting estimation scheme is often referred to as the \textit{pseudorandom method}. The ambiguity function of such a signal is illustrated in \cref{fig:pseudorand}.
 
\begin{figure}[H]
  \centering
     \includegraphics[width=0.6\textwidth]{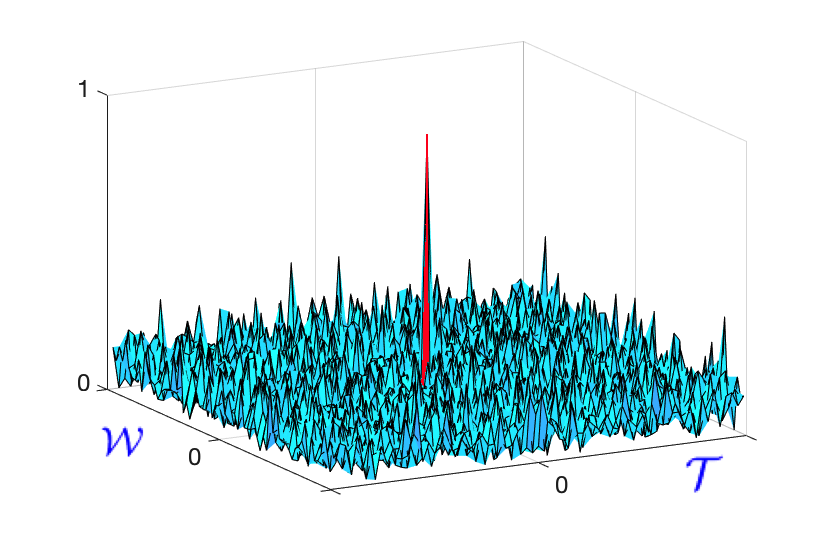}
  \caption{Ambiguity function of pseudorandom signal $S$.}
  \label{fig:pseudorand}
\end{figure}

The ambiguity function $\mathcal{A}$ also satisfies the following useful property.
\begin{prop}[Hermitian]\label{prop:amb1}
The ambiguity function $\mathcal{A}$ is a Hermitian form on $L^2(\Z_N)$. Namely, for $S_1,S_2,R \in L^2(\Z_N)$ and $\alpha \in \C$,
\begin{enumerate}
    \item $\mathcal{A}(\alpha \cdot S_1+S_2,R)= \alpha \cdot \mathcal{A}(S_1,R)+\mathcal{A}(S_2,R)$
    \item $\mathcal{A}(S_1,S_2)= \overline{\mathcal{A}(S_2,S_1)}$
\end{enumerate}

\end{prop}
This suggests the following estimation scheme for the pseudorandom method. 

\subsection{Estimation scheme}

We first briefly remind the reader that the digital estimation task (\ref{DET}) assumes \cref{eq:digital}. \\ The pseudorandom method can then be described as follows -- choosing and transmitting a pseudorandom signal $S$, we perform the digital estimation task (\ref{DET}) by locating the $k$ largest values of $|\mathcal{A}(S,R)|$. Note that this follows from \cref{prop:amb1} and \cref{prop:amb2}, under appropriate assumptions on noise and sparsity $k$. 

The authors of \cite{fish2014performance} demonstrate that the pseudorandom scheme works, in a sense which they make precise, under minimal assumptions on those parameters. \\

\vspace{0.7cm}
We next consider the arithmetic complexity of the described estimation scheme.

\subsection{Computational complexity}
Without any additional information, performing the above estimation algorithm  will involve estimating $N^2$ correlations, i.e., naively, this will require $O(N^3)$ arithmetic operations. \\
However, it is well-known \cite{fish2013delay} that on any line in the $\Z_N^2$ plane (for an exposition on lines in $\Z_N^2$, see \cref{app:lines}), for any two signals $S_1,S_2\in L^2(\Z_N)$, $\mathcal{A}(S_1,S_2)$ can be expressed as a convolution, as elaborated in \cref{app:ambigline}. The \textit{Convolution Theorem} and \textit{Fast Fourier Transform (FFT) Algorithm}\footnote{A relevant introductory exposition on the \textit{discrete fourier transform} (DFT) can be found in \cref{app:ft}.} together then imply that we can compute $\mathcal{A}(S_1,S_2)$ on a line in $O(N\log N)$ operations and, consequently, on the entire $\Z_N^2$ plane in $O(N^2\log N)$ operations.


\vspace{0.7cm}

In terms of arithmetic complexity, $O(N^2\log N)$ is certainly a significant improvement over $O(N^3)$, but with $N$ on the order of $10^6$, one might hope -- and expect -- that under realistic assumptions on noise and number of targets we can do still better. We now describe one such improvement.

\section{\textbf{Method of Chirps}}

In \cref{existing:prm} we picked a signal $S$ such that, with $R= H_{\tau_0,\omega_0}S$, the evaluations $|\<H_{\tau,\omega}S,R\>|$, for $(\tau,\omega)\in \Z_N^2$, have a distinguished maximum at $(\tau,\omega) = (\tau_0,\omega_0)$. 
A natural next consideration may be to pick the signal $S$ such that, for instance, the evaluations $|\<H_{\tau,0}S,R\>|$, for $\tau\in \Z_N$, have a distinguished maximum at $\tau = \tau_0$. The shift $\tau_0$ can then be detected by evaluating just $N$ correlations. In other words, we achieve half the goal with just $N$ correlations, which should produce an improvement over the arithmetic complexity of the pseudorandom method by a factor of $N$. \\
This was the basis for methods studied in \cite{fish2012delay,fish2013delay,fish2013incidence,fish2014}.

By \cref{prop:amb2}, a signal $S$ which satisfies the above requirement will have ambiguity $\mathcal{A}(S,R)$ with distinguished maxima all along on the shifted line $\set{(\tau_0,\omega)}:{\omega\in \Z_N}$. For convenience, let's adopt the following notation,
\begin{align*}
    \mathcal{W} &= \set{(0,\omega)}:{\omega\in\Z_N}.
\end{align*}
Equivalently, the ambiguity of $S$, $\mathcal{A}(S)$, would be essentially supported on $\mathcal{W}$.

\vspace{0.7cm}
Next, we consider how to pick signals that satisfy such specifications. 
\subsection{Chirp signals}\label{existing:chirps}
If one tried to think up functions in $L^2(\Z_N)$, the first example one might come up with are $\delta$-functions:
\begin{align*}
    \delta_{\tau}[\tau']&= \begin{cases}
        1 & \text{ if }\tau'=\tau \\
        0 &\text{ otherwise.}
    \end{cases},& \tau'\in\Z_N,
\end{align*}
for $\tau \in \Z_N$. \\
A quick check shows that the ambiguity of a $\delta$-function is supported on this line $\mathcal{W}$, as illustrated in \cref{fig:delta}. \\
After a little more thought, one might then think of the complex exponentials\footnote{The factor of $1/\sqrt{N}$ in the definition of $e_\omega$ simply ensures that the function has unit norm.}
\begin{align*}
    e_\omega[\tau] &= \frac{1}{\sqrt{N}}\cdot e^{\frac{2\pi i}{N}\omega\tau}, &\tau\in \Z_N, 
\end{align*}
for $\omega \in \Z_N$. The ambiguity of $e_\omega$ is supported on $\mathcal{T}$, for $\mathcal{T}$ given by
\begin{align*}
    \mathcal{T} &= \set{(\tau,0)}:{\tau\in\Z_N}.
\end{align*}
See \cref{fig:exp} for an illustration of the support of $e_\omega$, $\omega \in \Z_N$.

Moreover, $|\mathcal{A}(\delta_\tau,e_\omega)|=\frac{1}{\sqrt{N}}$. So, for $S=\delta_\tau+e_\omega$, the ambiguity $\mathcal{A}(\delta_\tau,S)$ is essentially supported on $\mathcal{W}$ by \cref{prop:amb1}, as shown in \cref{fig:deltaandexp} and similarly, $\mathcal{A}(e_\omega,S)$ will be essentially supported on $\mathcal{T}$.

In fact, for \textit{any} line $L = \set{(\tau,a\tau)}:{\tau\in\Z_N}$, where $a\in \Z_N$, we can expect functions $S_L$ whose ambiguity is supported on $L$. We will refer to elements of $L^2(\Z_N)$ whose ambiguity is supported on a line as \textit{chirp signals} or \textit{chirp functions}, or simply just as \textit{chirps}.  \\
\begin{figure}[H]
    \centering
        \begin{subfigure}[b]{0.8\textwidth}
        \includegraphics[width=\textwidth]{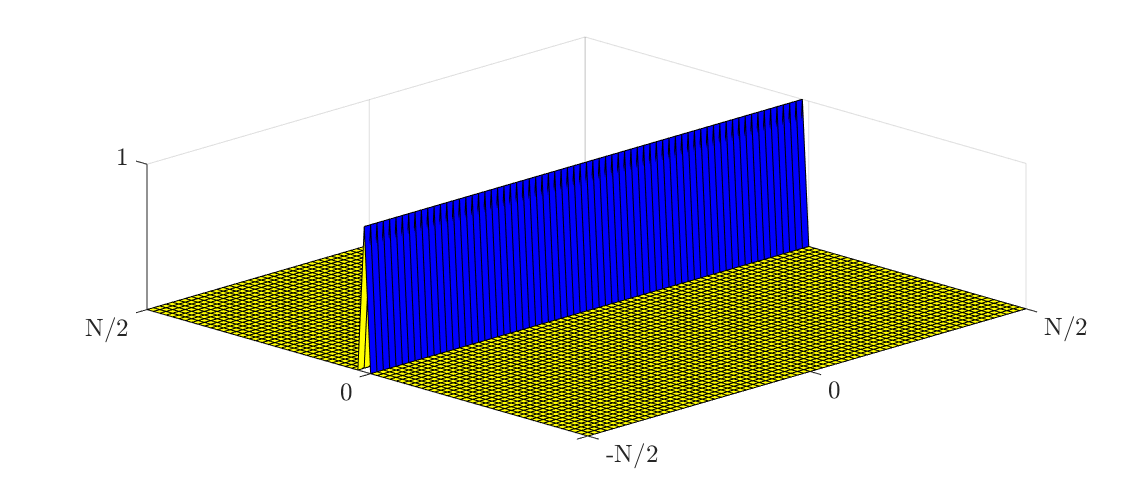}
        \caption{Ambiguity of a $\delta$-function, $\mathcal{A}(\delta_\tau)$.}\label{fig:delta}
        \end{subfigure}
        \begin{subfigure}[b]{0.8\textwidth}
        \includegraphics[width=\textwidth]{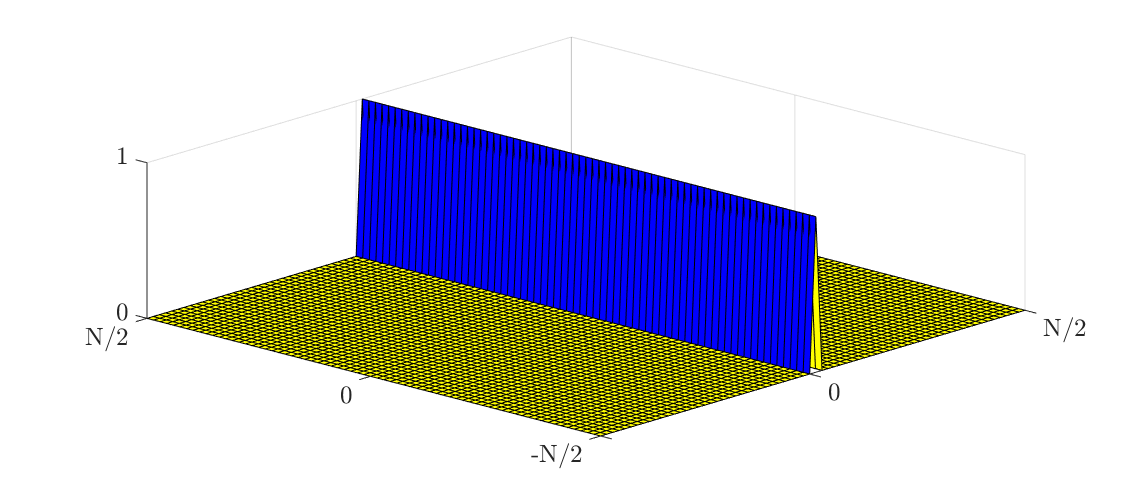}
        \caption{Ambiguity of a complex exponential, $|\mathcal{A}(e_\omega)|$.}\label{fig:exp}
        \label{fig:tshift1d}
        \end{subfigure}
        \begin{subfigure}[b]{0.8\textwidth}
        \includegraphics[width=\textwidth]{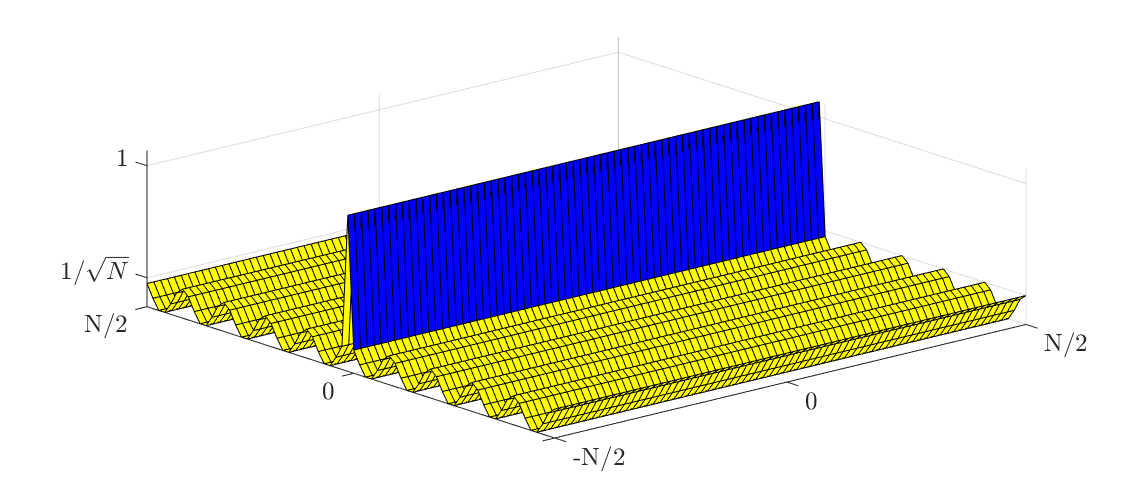}
        \caption{Real part of $\mathcal{A}(\delta_\tau,\delta_\tau+e_\omega)$.}
        \label{fig:deltaandexp}
        \end{subfigure}
        \caption{Visualizing various ambiguity functions involving $\delta_\tau$ and $e_\omega$.}
\end{figure}
The existence of such chirp functions is demonstrated in \cref{app:uaschirps}, \cref{thm:chirps}.  We construct explicit formulas for these signals (\cref{thm:chirpformula}) by utilizing the fact that operators $ \set{H_{\tau,a\tau}}:{\tau\in\Z_N}$, associated with the line $L$, can be seen as coming from special commuting subgroups the Heisenberg-Weyl group $G_{HW}$. Moreover, we show that for distinct lines $L\ne M$, $|\mathcal{A}(S_L,S_M)|=\frac{1}{\sqrt{N}}$ (\cref{lem:cross}). \\
While we don't present exact formulas for chirps at this point, we use the key fact, that they enable us to evaluate any sample $S_L[\tau]$, $\tau\in \Z_N$, in $O(1)$ arithmetic operations.

\vspace{0.5cm}
We will now detail a chirp-based estimation scheme. 


\subsection{Estimation scheme} 

For the sake of exposition, we begin this description with channel sparsity $k=1$.\\
The digital estimation process will involve choosing distinct lines $L, M$ and corresponding chirps $S_L, S_M$ and transmitting the signal $S =S_L+S_M$. The received signal $R$ will be given by \cref{eq:digital}. 
We can then locate $(\tau_1,\omega_1)$ as follows.
\begin{enumerate}
    \item Compute $|\mathcal{A}(S_L,R)|$, say, on the line $M$ and locate the distinguished peak. Let's denote the location of the peak as $(\tau_M,\omega_M)$,
    \item Compute $|\mathcal{A}(S_M,R)|$, say, on the line $L$ and locate the distinguished peak. Let's denote the location of the peak as $(\tau_L,\omega_L)$,
    \item Then $(\tau_1,\omega_1)$ is given by the point of \textit{double incidence} $(\tau_L,\omega_L)+(\tau_M,\omega_M)$.
\end{enumerate}
We now consider the case when channel sparsity $k>1$, for instance, $k=2$.\\
If we transmit, $S = S_L+S_M$, and follow the above steps to locate points of double incidence then, generically, rather than locating the two true shifts we will locate four points -- see \cref{fig:twolines} for illustration. However, we have some means of recourse to locate the true shifts. For instance,
\begin{enumerate}
    \item We might modify the estimation scheme to involve picking a third distinct line $K\subset \Z_N^2$ and corresponding chirp $S_K$, and transmit $S=S_L+S_K+S_M$ instead. As before, $R$ is given by \cref{eq:digital}. For a generic choice of lines $L,M,K$, the true shifts can be identified by points of \textit{triple incidence} (\cref{lem:genline}), as demonstrated in \cref{fig:threelines}. \\
    We can identify points of incidence three as follows: \\
    Points of double incidence can be located as before -- for instance, let $I_{LM}\subset\Z_N^2$ denote points of incidence two identified from $\mathcal{A}(S_L,R)$ and $\mathcal{A}(S_M,R)$. Similarly, we could also identify $I_{MK}$ and $I_{KL}$. Points of \textit{triple incidence} are given by their intersection -- in fact, the intersection of any two of them, e.g. $I_{LM}\cap I_{MK}$.\\
    This strategy is the basis of the \textit{Incidence Method}(IM) described in \cite{fish2013incidence,fish2013delay, fish2014}.
    \item True shifts could also be identified by using the values, rather than just the magnitude, of the ambiguity function at those points. This claim is based on the following observation: \\
    For generic $(\tau_1,\omega_1),(\tau_2,\omega_2)$, denoting $H=\sum_{j=1}^2\alpha_j\cdot H_{\tau_j,\omega_j}$, $(\tau,\omega) \in I_{LM}$ is a true shift if and only if,
    \begin{align*}
        \mathcal{A}(S_L,HS_L)[\tau,\omega] &= \mathcal{A}(S_M,HS_M)[\tau,\omega].
    \end{align*}
    The above observation coupled with the fact that, for instance, 
    \begin{align*}
        \mathcal{A}(S_L,R)[\tau,\omega] = \mathcal{A}(S_L,HS_L)[\tau,\omega] +O\p{\frac{1}{\sqrt{N}}}+ \text{Noise},    
    \end{align*}
    tells us that evaluations of $\mathcal{A}(S_L,R)$ and $\mathcal{A}(S_M,R)$ at $(\tau,\omega)\in I_{LM}$ can be used to identify true shifts.
    This strategy is the basis of the \textit{Cross Method} (CM) described in \cite{fish2012delay,fish2013incidence,fish2013delay, fish2014}. This method could be used independently or in conjunction with the Incidence Method (IM).
\end{enumerate}


\begin{figure}[H]
    \centering
    \begin{subfigure}[b]{0.48\textwidth}
        \includegraphics[width=\textwidth]{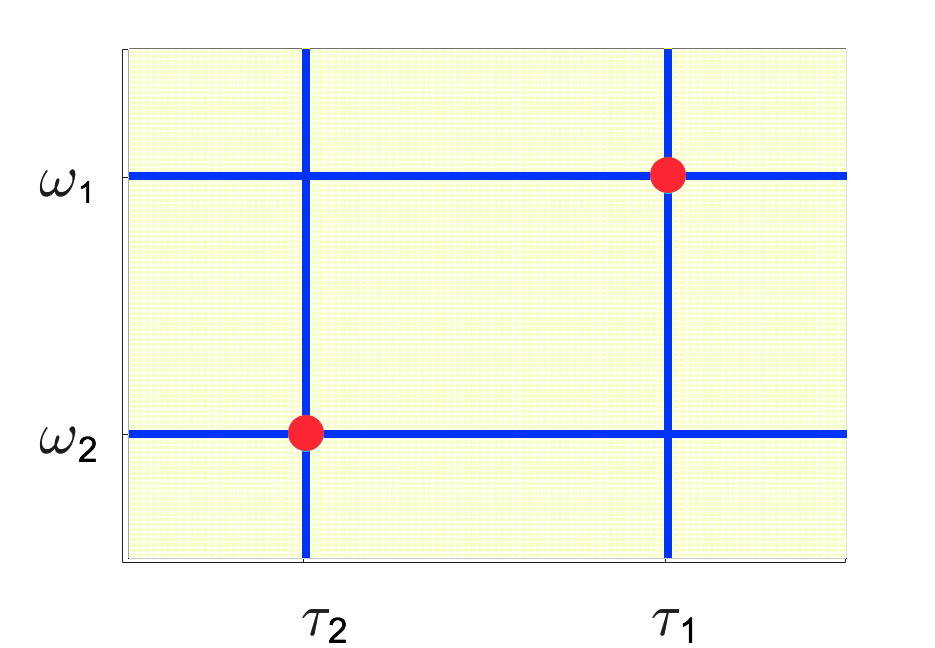}
        \caption{$S=S_\mathcal{W}+S_\mathcal{T}$, depicting essential support of $\mathcal{A}(S_\mathcal{W},R)$ and $\mathcal{A}(S_\mathcal{T},R)$ on the plane.}
        \label{fig:twolines}
    \end{subfigure}
    \begin{subfigure}[b]{0.48\textwidth}
        \includegraphics[width=\textwidth]{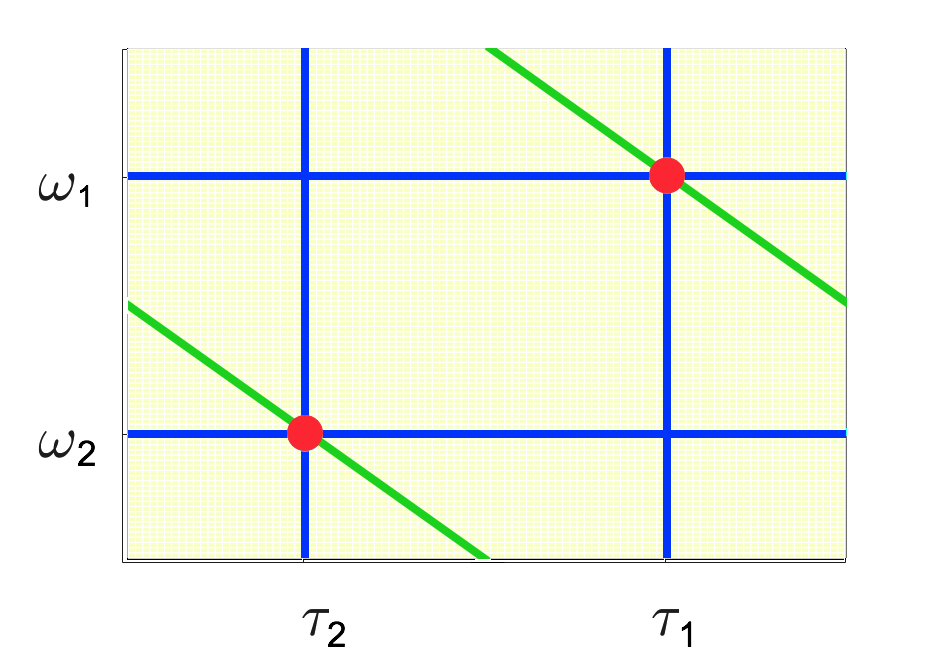}
        \caption{$S=S_\mathcal{W}+S_\mathcal{T}+S_K$, depicting essential supports of $\mathcal{A}(S_\mathcal{W},R)$, $\mathcal{A}(S_\mathcal{T},R)$, $\mathcal{A}(S_K,R)$.}
        \label{fig:threelines}
    \end{subfigure}
    \caption{Identifying true shifts $\Set{(\tau_1,\omega_1),(\tau_2,\omega_2)}$.}
\end{figure}

\vspace{0.5cm} \
Next, we demonstrate arithmetic complexity on the order of $N\log N + k^2$, for the described estimation scheme(s).

\subsection{Computational complexity}
In order to describe their computational complexity, we first note that the Incidence and Cross Methods both consist of an estimation algorithm involving two steps. As before, $S=S_L+S_M+S_K$, and $R$ will be given by \cref{eq:digital}. The arithmetic complexity of each step can then be described as follows:
\begin{enumerate}
    \item First, points of double incidence are located; for instance, $I_{LM}$ and $I_{MK}$. This can be achieved by evaluating, say, $\mathcal{A}(S_L,R)$ on $M$, $\mathcal{A}(S_M,R)$ on $L$ and $\mathcal{A}(S_K,R)$ on $M$, in a total of $O(N\log N)$ arithmetic operations.
    \item Next, points of triple incidence are located; for instance, $I_{LM}\cap I_{MK}$. For a generic choice of lines $L,M$, and $K$, both $I_{LM}$ and $I_{MK}$ consist of $k^2$ points, and so the intersection can be identified in $k^2$ operations.
\end{enumerate}
Altogether, this involves at most $c_1(N\log N+k^2)$ operations, $N$ samples of the received signal, and at most $ c_2N$ bits of storage, where $c_1$ and $c_2$ are constants independent of $N$ and $k$.

\vspace{0.7cm}

In terms of computational complexity, these methods are certainly an improvement over the pseudorandom method. However, for $N$ on the order of $10^6$ it would seem that for relevant practical levels of noise and number of targets, performing the digital estimation task (\ref{DET}) could be still more efficient.

For instance, in automotive systems bandwidth may be on the order of $GHz$ to estimate relative range with a desired resolution, however, the largest observed frequency shifts or \textit{doppler spread} can typically be assumed to be on the order of tens of $MHz$. \textit{Frequency Modulated Continuous Wave} or \textit{FMCW} radar \cite{jankiraman2018fmcw}, which is closely related to the method of chirps described here, exploits this fact about the doppler spread to work with fewer samples, in this case sampling on the order of $10^4$ rather than $10^6$ samples per millisecond. 

More interestingly, in some applications, it can also be assumed that the number of targets $k \ll N$. We refer to this as a \textit{$k$-sparse regime}, and  the digital estimation problem then becomes a \textit{$k$-sparse estimation problem}. Then, under reasonable assumptions on noise, one may hope to do still better than the above methods. This motivated the work in \cite{bar2014sub}, for instance, where the authors demonstrate an algorithm with sample complexity sublinear in $N$. 

In the next chapter, we present an estimation algorithm for the $k$-sparse regime with sampling, space \textit{and} arithmetic complexity sublinear in $N$.
%
\chapter{A Sublinear Algorithm}   
%
%
%
%
%
%

In this chapter, we describe a novel algorithm for channel estimation in the $k$-sparse regime, that we call \textit{Sparse Channel Estimation} or SCE for short. 

In order to explain the idea behind this algorithm, we first introduce another \textit{sparse estimation task}: \\
Consider $S\in L^2(\Z_N)$ such that,
\begin{align}\label{eq:sfft}
    S[\tau] & = \sum_{j=1}^k \alpha_j \cdot e_{\omega_j}[\tau] +\text{Noise},
\end{align}
where $\omega_j\in \Z_N$ and $k\ll N$. The coefficients $\alpha_j$ are closely related to the \textit{discrete fourier transform} of $S$, which we will denote as $\mathcal{F}S$ -- in particular, $\mathcal{F}S[\omega_j]$ is ``approximately" equal to $\alpha_j$. (For an introduction to fourier transforms, see \cref{app:ft}.) Since we have $k\ll N$, $\mathcal{F}S$ is \textit{essentially} $k$-sparse and the task of approximating $\mathcal{F}S$ is a \textit{sparse fourier estimation task}. This is made precise below.
\begin{task}[\textbf{\textit{Sparse fourier estimation task}}]\label{SET}
Detect $\omega_j \in \Z_N$ and estimate $\alpha_j$, for $j = 1, \dots, k$.
\end{task}

We are only interested in estimation schemes that perform the sparse fourier estimation task (\ref{SET}) with arithmetic, sampling and storage complexity at most \textit{sublinear} in $N$. Under certain assumptions on noise, it is known that this task can be performed with sampling and arithmetic complexity\footnote{Typically, the storage complexity of these schemes is simply proportional to sampling complexity, and so we may neglect to mention storage from here on out.} sublinear in $N$, and there is a large body of work on algorithms that achieve this \cite{gilbert2014recent,indyk2014nearly}. They are typically referred to as \textit{Sparse Fast Fourier Transform} or SFFT algorithms.

To see why the above discussion is relevant to our digital estimation scheme, consider the method of chirps. We will first observe that for a chirp signal $S_L$, and for any $R\in L^2(\Z_N)$, the ambiguity function $\mathcal{A}(S_L,R)$ on a line $M\ne L$ is, loosely speaking, a fourier transform of $R$ (this statement is made precise in \cref{thm:redtosfft}). It then follows that in the $k$-sparse regime, for $S=S_L$ and $R$ given by \cref{eq:digital}, the support of $\mathcal{A}(S,R)$ on $M$ is essentially $k$-sparse, and so estimating $\mathcal{A}(S,R)$ on this line becomes a sparse fourier estimation task (\ref{SET}). This is the basis of the Sparse Channel Estimation or SCE method.

\vspace{0.7cm}
In order to effectively describe Sparse Channel Estimation (SCE), we first overview the ``moving parts" of a typical SFFT algorithm.

\section{An overview of Sparse Fast Fourier Transform (SFFT) algorithms.}

Any SFFT algorithm will involve a randomized subsampling scheme and estimation procedure that utilizes the evaluated samples. Sparse Channel Estimation (SCE) does not assume a particular choice of SFFT algorithm -- SCE is \textit{modular} in this sense, and one can ``plug-in" the SFFT algorithm of their choice. \\
So, while we do not describe a specific sparse fourier transform algorithm here, we will summarize some of standard components of these methods. We describe those components by considering a couple of special cases of the sparse fourier estimation task (\ref{SET}), as we work our way up to the general case.
\begin{enumerate}
    \item \textbf{\textit{Case: $k\le1$} -- the $1$-sparse algorithm.} \\
    If $k=1$, then without noise,
    \begin{align*}
        S[\tau] = \alpha_1 \cdot e^{\frac{2\pi i}{N}\omega_1\tau}
    \end{align*}
    In this case, $\omega_1$ can be estimated from the argument of $S[1]/S[0]$, and $\alpha_1$ is simply given by $S[0]$. In other words, $\mathcal{F}S$ can be estimated with a deterministic sampling scheme of just two samples! 
    
    Returning to a model with noise, namely, (\ref{eq:sfft}) with $k\le 1$, we can perform \cref{SET} with a ``robustified" version of the above procedure. This version would involve a randomized subsampling scheme, and corresponding estimation procedure. Since the sparsity $k$ could be $1$ or $0,$ the estimation procedure comes with a positive, real-valued \textit{thresholding parameter} $\mu$ -- if the estimated coefficient $\widehat{\alpha}_1$ has magnitude less than $\mu$, $k$ is declared to be zero. For the reader's convenience, we will denote this procedure as $1\text{SFFT}_\mu$. 
    
    The theme of modularity continues here as well, since SCE does not assume a specific choice of $1$-sparse algorithm and can be implemented with a method of choice. Examples of such methods can be found in \cite{gilbert2014recent,indyk2014nearly,pawar2014robust}, for instance. However, in order to provide statistical guarantees and bounds on computational complexity, we make a choice of $1$-sparse algorithm, see \cref{app:1sparse} for the same. This algorithm takes about $\log N$ samples and arithmetic operations (\cref{lem:bitbybit,lem:thresh}).
    \item \textbf{\textit{Case: $k>1$, $\omega_j$ uniformly spread -- Discrete filter functions.}}\\
    \begin{figure}[]
    \centering
    \begin{subfigure}[b]{0.7\textwidth}
        \includegraphics[width=\textwidth]{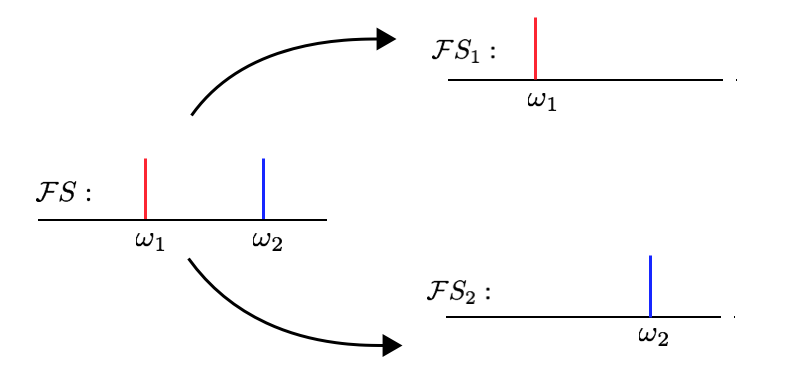}
        \caption{Schematic representation of reduction to $1$-sparse case.}
        \label{fig:schemek2}
    \end{subfigure}
    \begin{subfigure}[b]{0.7\textwidth}
        \includegraphics[width=\textwidth]{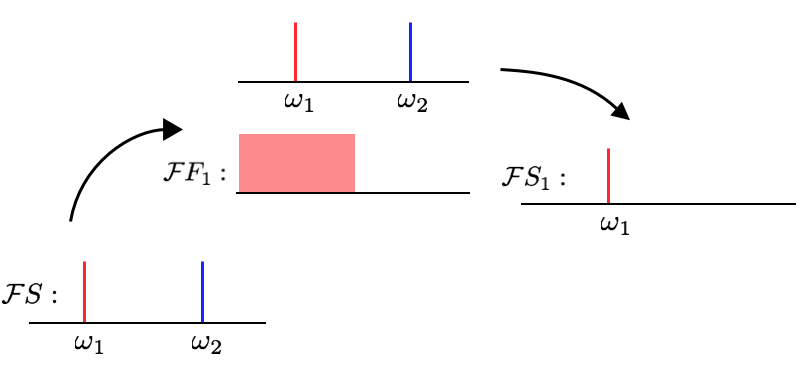}
        \caption{Schematic representation of filtering to reduce to $1$-sparse case.}
        \label{fig:filterscheme}
    \end{subfigure}
    \caption{$k=2$, $\omega_1$ and $\omega_2$ are well spread.}
    \label{fig:filtering}
\end{figure}
    We use the term \textit{uniformly spread} here to mean that each of the $k$ intervals,
    \begin{align*}
    [0,N/k], [N/k,2N/k], \dots, [(k-1)N/k,0] \subset \Z_N,
    \end{align*}
    contains only one $\omega_j$ for some $j=1, \dots, k$. In this case, \cref{SET} will be performed by reducing to the $1$-sparse case, see \cref{fig:schemek2} for a schematic illustration of this. This reduction will be achieved with certain elements of $\L^2(\Z_N)$ called \textit{filter functions}, or simply \textit{filters}. There are two prescriptions for such functions, which we enumerate below. For ease of exposition, we illustrate these requirements for $k=2$:
    \begin{itemize}
        \item We will require two filter functions $F_1$ and $F_2$, such that $\mathcal{F}F_1$ and $\mathcal{F}F_2$ are supported on intervals of length $N/2$. For instance, $\mathcal{F}F_1$ may be supported on the interval $[0,N/2]\subset{\Z_N}$, and $F_2$, given by a frequency shift,
    \begin{align*}
        F_2 = e_{N/2}\cdot F_1,
    \end{align*}
    will have $\mathcal{F}F_2$ supported on the interval $[N/2,N]\subset{\Z_N}$. 
    
    Now consider $F_1*S$ and, for simplicity, let us denote it as $S_1$. By the \textit{Convolution Theorem}, 
    \begin{align}
        \mathcal{F}(S_1)&= \mathcal{F}(F1 * S)= \mathcal{F}F_1\cdot \mathcal{F}S.
    \end{align}
     So, as illustrated in \cref{fig:filterscheme}, $\mathcal{F}(S_1)$ is $1$-sparse as desired, and so we apply the randomized sampling scheme and estimation procedure, $1\text{SFFT}_\mu$, from Case 1 here.
    
    However, for any $\tau\in \Z_N$, a single sample $S_1[\tau]$ is given by,
    \begin{align*}
        S_1[\tau] & = \sum_{\tau'\in \Z_N} F_1[\tau']\cdot S[N-\tau'].
    \end{align*}
    In other words, the number of arithmetic operations and samples of $S$ required to evaluate a single sample of $S_1$ is determined by the size of $F_1$'s support. This brings us to the second prescription.
    \item We may wish to enforce that the size of the support of $F_1$ (and, consequently, of $F_2$) be sublinear in $N$. However, the \textit{Fourier Uncertainty Principle} tells us that no function can satisfy both the above requirements simultaneously! 
    
    Fortuitously, there exist functions in $L^2(\Z_N)$ that ``approximately" satisfy both requirements. The use of the term ``approximately" here is made precise in \cref{defn:filter,defn:esssup}. More precisely, there is a family of functions $\Set{F^k}_{k\le N}$, such that $F^k$ is essentially supported on an interval of length $k\log N$, and $\mathcal{F}F^k$ is essentially supported on an interval of length $N/k$, see \cref{thm:filters} for a construction of the same. The existence of such a family of functions is mathematically non-trivial -- it follows from the existence of an eigenfunction of the continuous fourier transform that is ``highly localized", namely, the standard Gaussian.  
    
    If we now denote, 
    \begin{align*}
    S_j &= F_j^k*S &\text{ for }j=1, \dots,k,
    \end{align*}
    where $F_1^k= F^k$ then, by construction of $F^k$, we have reduced to Case 1 with each $S_j$. Moreover, for any $\tau\in \Z_N$ we can \textit{simultaneously} evaluate $S_j[\tau]$, $j=1, \dots,k$, with roughly $k$ operations and $k$ samples of $S$ (\cref{lem:filtering}). We would then like to simultaneously evaluate samples of $S_1, \dots, S_k$ using the randomized sampling scheme from Case $1$ and this, in effect, produces a subsampling scheme for $S$. 
    
    Once $S$ has been sampled, we will perform an estimation procedure which involves evaluating samples of $S_j$ and applying $1\text{SFFT}_\mu$, for $j=1, \dots, k$. For the reader's convenience, we will refer to the above procedure simply as $\text{SFFT}_\mu$.
     
    \end{itemize}

    \item \textbf{\textit{General case -- Pseudorandom spectral permutation.}} \\ 
    \textit{Pseudorandom spectral permutation} provides a means of sampling $S$ by which, intuitively speaking, we can assume that the $k$ frequencies $\omega_j$ are uniformly spread in $\Z_N$. In addition, it comes without additional cost to sampling or arithmetic complexity. 
    
    Pseudorandom spectral permutation is based on the following idea. \\
    Let $\Z_N^\times$ denote the set (in fact, group) of elements in $\Z_N$ with a multiplicative inverse modulo $N$. An element $\sigma\in \Z_N^\times$ acts on $L^2(\Z_N)$ by \textit{scaling} as follows:
    \begin{align}\label{eq:permute}
        \sigma: S&\mapsto S^{\sigma}, \\
        S^{\sigma}[\tau] & = S[\sigma \cdot \tau],
    \end{align}
    for any $S \in L^2(\Z_N)$. 
    
    This \textit{action} commutes with the discrete fourier transform as follows,
    \begin{align*}
        \mathcal{F}S^{\sigma} &= (\mathcal{F}S)^{\sigma^{-1}},\\
        \mathcal{F}S^{\sigma}[\omega] &= \mathcal{F}S[\sigma^{-1}\omega]
    \end{align*}
    Thus, scaling by a random element of $\Z_N^\times$ will randomly permute $\omega_1, \dots \omega_k$, as illustrated in \cref{fig:permute}.
    \begin{figure}[h]
        \begin{subfigure}[b]{1\textwidth}
            \centering
           \includegraphics[width=.8\textwidth]{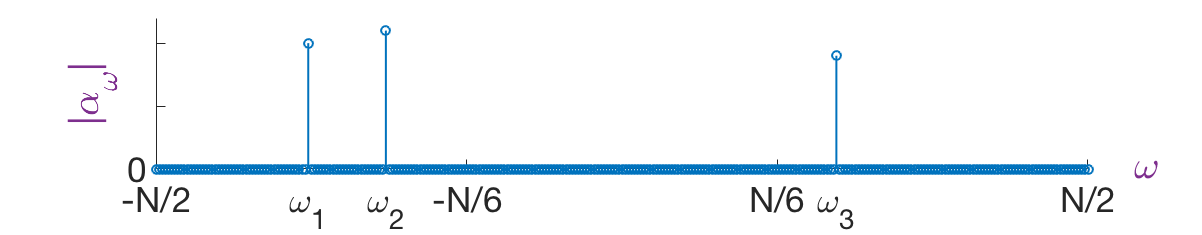}
             \caption{$\omega_j $ are not uniformly spread.}
             \label{fig:filtime} 
        \end{subfigure}
        
        \begin{subfigure}[b]{1\textwidth}
            \centering
           \includegraphics[width=.8\textwidth]{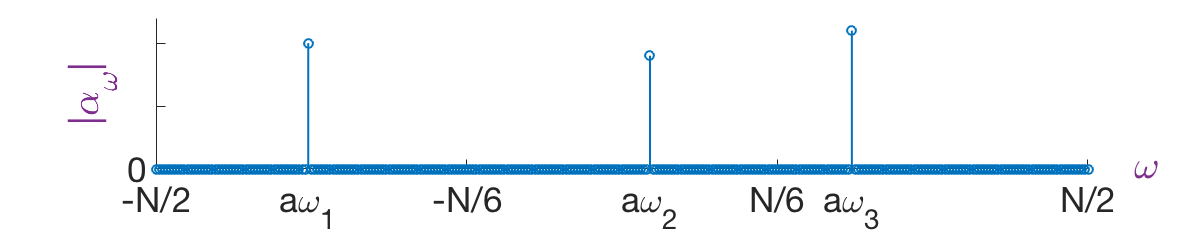}
             \caption{$a=163$, $a\omega_j $ are uniformly spread.}
             \label{fig:filfreq}
        \end{subfigure}
        \caption{Reducing to Case 2 using a pseudorandom permutation, $\omega_1 = -100$, $\omega_2 = -75$, $\omega_3=70$, $N=300$.}
        \label{fig:permute}
    \end{figure}
    Permutation by random scaling can be further ``augmented" with a frequency shift by $a\in \Z_N$ picked at random,
    \begin{align*}
        (\sigma,a)&: S\mapsto S^{\sigma,a}, \\
        S^{\sigma,a}[\tau] & = e_a[\tau]\cdot S[\sigma \cdot \tau], \\
        \mathcal{F}(e_a[\tau]\cdot S^{\sigma})[\omega] &= \mathcal{F}S[\sigma^{-1}\omega-a]
    \end{align*}
    for any $S \in L^2(\Z_N)$. 
    
    \cref{cor:isolate} tells us that with about $\log N$ random spectral permutations we can, loosely speaking, reduce the general case to Case 2. For each of those permutations, $(\sigma,a)$, we would then like to sample $S^{\sigma,a}$ using the sampling scheme from Case 2. This, in effect, produces a sampling scheme for $S$ which we will refer to as $\text{Sampling}_{\text{SFFT}}$.
    
    Once $S$ has been sampled, we will then have an estimation procedure which involves evaluating samples of $S^{\sigma,a}$ and applying $\text{SFFT}_\mu$ in order to perform \cref{SET}. Despite a minor modification, we will continue to refer to this estimation procedure as $\text{SFFT}_\mu$, for ease of exposition.

    
    
\end{enumerate}

\vspace{0.7cm}
With this overview of SFFT algorithms, we are now ready to describe the proposed method.

\section{Sparse Channel Estimation (SCE)}

The Sparse Channel Estimation (SCE) scheme employs a method of chirps, and so we transmit $S= S_L+S_M+S_K$ for three distinct randomly chosen lines $L,M,K$. $R$ is given by the digital channel model (\ref{eq:digital}) but in addition, SCE presumes the $k$-sparse regime.

We now describe the estimation algorithm. \\
The method of chirps would then involve estimating $\mathcal{A}(S_L,R)$, $\mathcal{A}(S_M,R)$ and $\mathcal{A}(S_K,R)$ on certain lines in the plane.
Recall that in the $k$-sparse regime, evaluating, for instance, $\mathcal{A}(S_L,R)$ on a line reduces to performing \cref{SET}, in this case (by \cref{thm:redtosfft}), on $S_L\cdot\overline{R}$, and so we can use an SFFT algorithm. From the previous discussion, the first step in such an algorithm would be to apply $\text{Sampling}_{\text{SFFT}}$ to $S_L\cdot\overline{R}$. This will, in effect, produce a sampling scheme for $R$. \\Once $R$ is sampled, we can then evaluate samples of $S_L\cdot\overline{R}$, $S_M\cdot\overline{R}$ and $S_K\cdot\overline{R}$, to estimate the respective ambiguities. For the reader's convenience, we denote the set of samples of $S_L\cdot \overline{R}$ as $\text{Samples}_L$ and, similarly, we also have $\text{Samples}_M$ and $\text{Samples}_K$. We will refer to this process as $\text{Sampling}_{SCE}(R)$.

We can then apply $\text{SFFT}_\mu$ to $\text{Samples}_L$, for instance, in order to estimate $\mathcal{A}(S_L,R)$ on any line (distinct from $L$) in the plane $\Z_N^2$. \cref{thm:redtosfft} provides the explicit formula that relates the output of $\text{SFFT}_\mu$ and peaks of $\mathcal{A}(S_L,R)$ on a line in $\Z_N^2$.

\vspace{0.7 cm}

We present the pseudocode for SCE below. 

\begin{algorithm}[ht]
\caption{Sparse Channel Estimation ($\text{SCE}_\mu$)}\label{alg:SCE}
\begin{algorithmic}[1]
\renewcommand{\algorithmicrequire}{\textbf{Input:}}
 \renewcommand{\algorithmicensure}{\textbf{Output:}}
\REQUIRE
Channel sparsity $k$.
\STATE Randomly choose transversal lines $K,L,M$
\STATE Transmit $S=S_K+S_L+S_M$. 
\STATE $\text{Sampling}_{SCE}(R)$ $\rightarrow$ $\text{Samples}_L, \text{Samples}_M, \text{Samples}_K$.
\STATE Locate peaks of $\mathcal{A}(S_K,R)$ on $L$ using $\text{Samples}_K$ and $\text{SFFT}_\mu$  \\ \hspace{5mm} $\rightarrow\kappa_i \in \Z_N^2$, $i \in \Set{1, \dots k}$. \\
\STATE Locate peaks of $\mathcal{A}(S_L,R)$ on $K$ using $\text{Samples}_L$ and $\text{SFFT}_\mu$\\ \hspace{5mm} $\rightarrow \ell_j \in \Z_N^2$, $j \in \Set{1, \dots k}$. 
\STATE Locate peaks of $\mathcal{A}(S_{M},R)$ on $L$ using $\text{Samples}_K$ and $\text{SFFT}_\mu$\\  \hspace{5mm}$\rightarrow m_{j'}\in \Z_N^2$, $j' \in \Set{1, \dots k}$. 
\STATE Find points of triple incidence,  $$\Set{\kappa_i+\ell_j}\cap \Set{\kappa_{i}+m_{j'}}.$$
\ENSURE Return every shift $(\tau, \omega)$ found in Step 7.
\end{algorithmic}
\end{algorithm}
%
\chapter{Complexity Bounds and Guarantees for Sparse Channel Estimation (SCE)}   
%
%
%
%
%
%

We now analyze the performance of SCE. We do so using certain standard measures to quantify the quality of such an algorithm, namely, \textit{probability of detection (PD)} and \textit{probability of false alarm (PFA)} \cite{kay1993fundamentals}.
\begin{definition} The \underline{probability of detection} (PD) of a channel estimation scheme is the probability that the $j^{th}$ target, $(\tau_j,\omega_j)$, is estimated.
\end{definition}

\begin{definition}The \underline{probability of false alarm} (PFA) of a channel estimation scheme is the probability that the $j^{th}$ shift estimated, does not correspond to any target.
\end{definition}

The guarantee for \textit{Sparse Channel Estimation} (SCE), \cref{thm:SCE}, assumes the model \cref{eq:digital} together with the following additional features. 
\begin{itemize}
\item We make an assumption on the channel sparsity.

\textbf{A1} (\textit{Sparsity}): The channel sparsity is at most $k$, where $k\ll N$ is a constant, i.e., independent of $N$.

\item We also make the following assumption on the coefficients $\alpha_j$ and their distribution.

\textbf{A2} (\textit{$\epsilon$-targets}):
There is some $A>0$ and $\epsilon\in (0,1)$ such that $(\alpha_1, \dots,\alpha_k)$ is drawn uniformly at random from the following set:
\begin{align*}
    B_{\epsilon} & = \set{x\in \C^k}:{\|x\|^2 = A \text{ and } \min_{x_j\ne 0} |x_j|\ge \epsilon\cdot \sqrt{A/k}}.
\end{align*}

\item The final assumption that we make is on the distribution of the noise $\nu_n$.

\textbf{A3} (\textit{Subgaussian}): We assume 
$\nu_\tau$ are i.i.d, mean zero and \textit{subgaussian} random variables \cite{vershynin2018high} with subgaussian parameter $\sigma^2/N$. 
\end{itemize}
For our purposes, we define \textit{signal-to-noise ratio}, or \textit{SNR} for short, to be $\text{SNR} = A/\sigma^2$. 

\begin{restatable}[SCE]{theorem}{thmSCE}\label{thm:SCE} 
Let $\mu = \kappa\cdot \epsilon\sqrt{A/k}$ for some confidence parameter $\kappa\in (0,1)$. \\ Then, under the sparsity, $\epsilon$-targets, and subgaussian assumptions there is an implementation of $\text{SCE}_\mu$ which takes
\begin{enumerate}
    \item $c_1k (\log N)^3(\epsilon^{2}\text{SNR} )^{-1}$ samples,
    \item $c_2 k (\log N)^3(\epsilon^{2}\text{SNR})^{-1}$ bits of memory, and
    \item  $c_3 k (\log N)^3(\epsilon^{2}\text{SNR})^{-1}+k^2$ arithmetic operations 
\end{enumerate}
for which $\textit{PD} \to 1$ and $\textit{PFA}\to 0$ as $N\to \infty$, where $c_1,c_2,c_3$ are constants independent of $\epsilon$, $\textit{SNR}$, $k$ and $N$. 
\end{restatable}
\begin{remark} 
A3 is satisfied, for instance, by the standard assumption of additive white Gaussian noise (AWGN) \cite{tse2005fundamentals}, namely, that $\nu_\tau$ are i.i.d, mean zero and \textit{Gaussian} random variables with variance $\sigma^2/N$. 
\end{remark}
\begin{remark} The proof of Theorem \ref{thm:SCE} confirms that for an appropriate choice of constants $c_1,c_2,c_3$, the rate of convergence of \mbox{$\textit{PD} \to 1$} and $\textit{PFA} \to 0$ is at least polynomial in $N$. 
\end{remark}
%
\chapter{Numerical Results}\label{sec:numres}   
%
%
%
%
%
%

We now provide the experimentally observed convergence rates for $\text{PD}$ and $\text{PFA}$ in a specific case, see \cref{fig:PD,fig:PFA}.

In addition, we provide numerical comparisons for time and space complexity of SCE and the Incidence method (IM) in a specific case, see \cref{tab:my_label}. 
\begin{figure}[H]
  \centering        \includegraphics[width=0.66\textwidth]{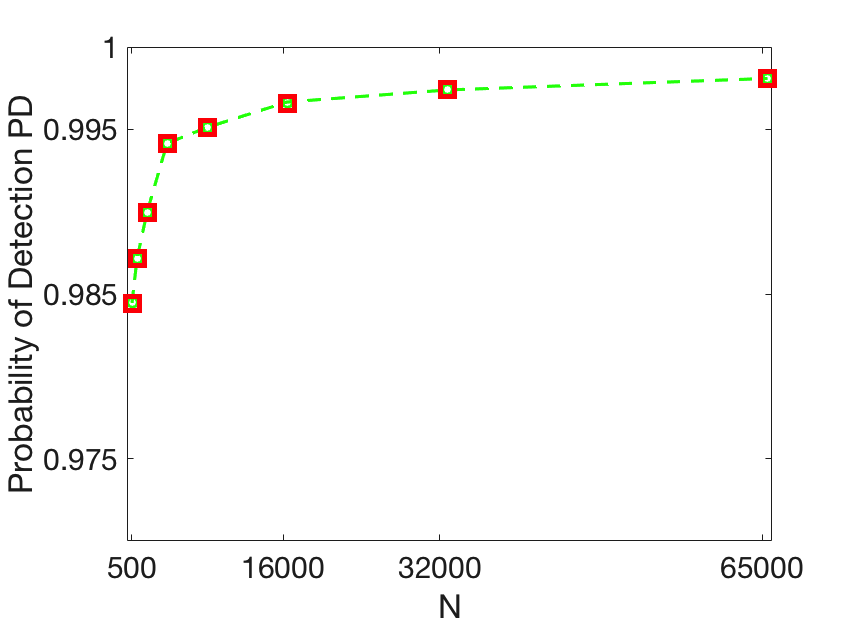}
        \caption{Experimentally observed convergence rate for PD, \mbox{$k=5$}, $\text{SNR}=10$dB, $500$ random trials.}
  \label{fig:PD}
  \vspace{-.3in}
\end{figure}
\begin{figure}[htp]
  \centering
    \includegraphics[width=0.66\textwidth]{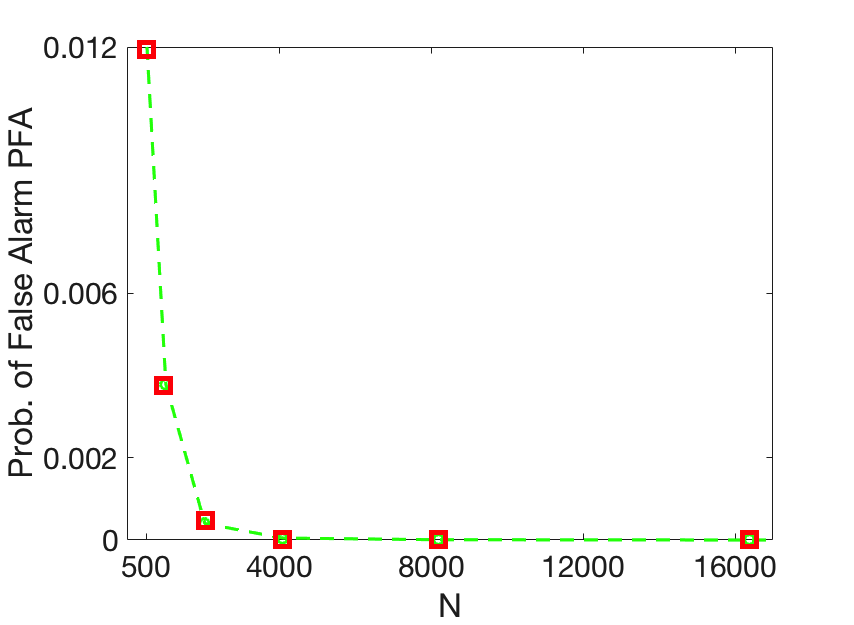}
    \caption{Experimentally observed convergence rate for PFA, \mbox{$k=5$}, $\text{SNR}=10$dB, $500$ random trials.}
    \label{fig:PFA}
    \vspace{.05in}
\end{figure}
\begin{table}[H]
    \centering
    \begin{tabular}{|c|c|c|c|c|}
    \hline
         &\multicolumn{2}{| c |}{SCE} &\multicolumn{2}{| c |}{IM}  \\
         \hline 
         N & Samples & Time (sec) & Samples & Time (sec)\\
    \hline 
   2048  &  2048  & 0.1423& 2048 &    0.0370 \\
   4096   & 4096 &  0.1250 & 4096 & 0.0740 \\
    8192   & 5468   & 0.1595  &  8192 & 0.1620 \\
     16,384   & 6242 & 0.1789 &   16,384  &  0.3230  \\
        32,768 & 7064 & 0.1992 &  32,768 & 0.5790    \\
       65,536  & 7934  & 0.2314 &  65,536   &  1.2010 \\
    \hline
    \end{tabular}
    \caption{Numerical comparison of time and sample complexity, $k=50$, $\text{SNR}=10$dB, $500$ random trials.}
    \label{tab:my_label}
    \vspace{-.15in}
\end{table}

\appendix       
\chapter{Justifications behind the Model.}

\section{The Doppler effect as a time-scale.}\label{app:timescale}

\begin{reusefigure}[H]{fig2}
  \centering    
     \includegraphics[width=0.7 \textwidth]{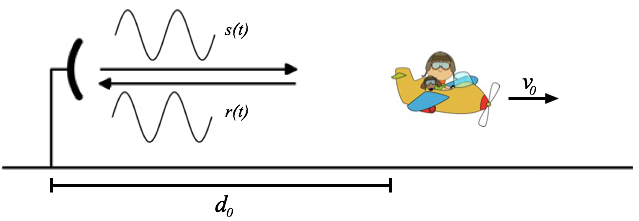}
     \caption{Illustrating range and radial velocity of object.}
     \vspace{-.15in}
\end{reusefigure}

As in \cref{sec:cont}, if the object of interest is at a radial distance of $d_0$, and has relative radial velocity $v_0$, then there is a standard model \cite{levanon2004radar} for the relationship between the signal transmitted $s$ and the signal received $r$, namely,
\begin{align}
    r(t)& =  \upalpha_0 \cdot s\p{(a_0)t-t_0}+ \text{Noise}  
\end{align}
where,
\begin{align*}
    a_0 &= 1-\frac{v_0}{v_0+c} \\
    t_0 &= \frac{2d_0}{c}.    
\end{align*}

For a detailed justification, see Section 1.1 in \cite{levanon2004radar}.

However, it in many applications, bandwidth W is significantly less than the carrier frequency $f_c$. In this case, it is standard to make the following approximation,
\begin{align}\label{eq:tstofs}
  s(a_0t) \approx e^{if_0t}s(t),  
\end{align}
where $f_0 = f_c|a_0-1|$, which leads to the \textit{continuous channel model} \ref{eq:cont}. \\
Next, we provide a mathematical basis for making the above approximation \ref{eq:tstofs} when $W/f_c\ll 1$.

\subsection{Time scale to frequency shift.}\label{app:tstofs}

We have the following statement that relates a time-scale of a signal $s$ to a frequency shift of the same, and involves the ratio $W/f_c$ of its bandwidth to carrier frequency. 
\begin{restatable}[]{theorem}{narrowband}\label{thm:narrowband}
For $s \in L^2(\R)$ with bandwidth $W$ and carrier frequency $f_c$,
\begin{align}
s(a_0t) = e^{if_0t}s(t)+ O\p{f_0t\cdot\frac{W}{f_c}}
\end{align}
where 
$f_0 = f_c|a_0-1|$.
\end{restatable}


\begin{remark}\label{rem:narrowband} It seems to be a standard rule of thumb to make the approximation \ref{eq:tstofs} if the signal bandwidth is less than one-tenth of the carrier frequency \cite{levanon2004radar}. In Section 1.1 of \cite{levanon2004radar}, the authors remark, ``Numerical simulations with rather complicated signals showed that the difference between the calculated performances was very small, even when the narrowband assumption was used with a signal whose bandwidth reached 40\% of the center frequency." They also note that Appendix A of \cite{di1968radar} lists errors resulting from these approximations.
\end{remark}

\section{Moving between Digital and Analog Settings.}\label{app:digital-analog}

Shannon, in his seminal work \cite{shannon1949communication} provided explicit formulas that can be recast as a linear map,
\begin{align*}
    \text{D-to-A}&: L^2(\Z_N) \to L^2(\R), \\
    S&\mapsto s,
\end{align*}
that would produce signals of duration $T$ and bandwidth $W$. Namely,
for $S\in L^2(\Z_N)$, let $s=\text{D-to-A}(S)$; then,
\begin{align}\label{eq:daformula}
    s(t) &= e^{2\pi if_ct}s_0(t)
\end{align}
where,
\begin{itemize}
    \item $s_0(t) = \sum_{\tau\in \Z_N} S[\tau]\cdot \text{sinc}_W(t-\frac{\tau}{W})$,
    \item $\text{sinc}_W(t)  = \frac{\sin(\pi Wt)}{\pi Wt}$,
    \item $f_c$ denotes the carrier frequency of the signal to be transmitted, and $f_c$ is a multiple of $W$.
\end{itemize}
\ 
In addition, he provided exact formulas that realize the following  linear map,
\begin{align*}
   \text{A-to-D}&: L^2(\R) \to L^2(\Z_N), \\
   r &\mapsto R
\end{align*}
Namely, for $r\in L^2(\R)$, let $R = \text{A-to-D}(r)$; then,
\begin{align}\label{eq:adformula}
R[\tau] &= \sum_{m\in \Z} r\p{\frac{\tau}{W}+mT} 
\end{align} 

\cref{prop:ad-da}, which is the key attribute of $\text{D-to-A}$ and $\text{A-to-D}$ from which the digital channel model \ref{eq:digital} is derived, can be immediately derived from the above formulas. For a proof, see \cref{appD}. 

\chapter{Underlying Algebraic Structure.}
\section{Computing the Ambiguity function on a Line}\label{app:ambigline}

Let $L\subset\Z_N^2$ be a line generated by  $(1,a)$ for some $a\in \Z_N$. Then,
\begin{align*}
    \mathcal{A}(S_1,S_2)(\tau,a\tau+\omega)&= (\widetilde{S_1}* \widetilde{S_2})[\tau],
\end{align*}
where,
\begin{align*}
    \widetilde{S_1}[\tau]&= e^{\frac{2\pi i}{N}\p{-a\tau^2-\omega\tau}}\cdot S_1\\
    \widetilde{S_2}[\tau]&= e^{\frac{2\pi i}{N}a\tau^2}\cdot \overline{S_2[-\tau]}.
\end{align*}
If $L$ is generated by $(0,1)$ then \cite{fish2013delay},
\begin{align*}
    \mathcal{A}(S_1,S_2)(\tau,\omega)&= \sqrt{N}\cdot(\mathcal{F}S_1* \mathcal{F}\widetilde{S_2})[\omega]\\
    &= \sqrt{N}\cdot\p{\mathcal{F}(S_1\cdot \widetilde{S_2})}[\omega]
\end{align*}
where,
\begin{align*}
    \widetilde{S_2}[\tau']&= \overline{S_2[\tau'-\tau]}.
\end{align*}
The above formulas can be verified by direct evaluation. A verification can also be found in \cite{fish2013delay}.
\section{Chirp Signals}
\label{app:uaschirps} 

In \cref{existing:chirps}, we saw that the ambiguity of a $\delta$-function $\mathcal{A}(\delta_\tau)$, was supported on the line $\mathcal{W}$. We also saw that for a complex exponential $e_\omega$, $\mathcal{A}(e_\omega)$ is supported on the line $\mathcal{T}$, and claimed that for \textit{any} line $L = \set{(\tau,a\tau)}:{\tau\in\Z_N}$, where $a\in \Z_N$, we can expect chirp functions whose ambiguity is supported on $L$. To see why, we point out that
\begin{itemize}
    \item $\set{\delta_\tau}:{\tau\in \Z_N}$ forms an orthonormal basis for the commuting operators \\ $\set{H_{0,\omega}}:{\omega\in \Z_N}$ associated with the line $\mathcal{W}$ and,
    \item Likewise, $\set{e_\omega}:{\omega\in \Z_N}$ forms an orthonormal basis for the commuting operators $\set{H_{\tau,0}}:{\tau\in \Z_N}$ associated with the line $\mathcal{T}$.
\end{itemize} 
In the same way, the operators $\set{H_{\tau,a\tau}}:{\tau\in \Z_N}$, associated with the line $L$, commute and share an orthonormal basis which we will denote as $\set{S_L^\omega}:{\omega\in \Z_N}$. It then turns out that $\mathcal{A}(S^\omega_L)$ is supported on $L$. This result (\cref{thm:chirps}) follows immediately from the \textit{Heisenberg commutation relation}, namely that,
\begin{align}\label{eq:commutation}
    H_{\tau,\omega}\circ H_{\tau',\omega'}&= e^{\frac{2\pi i}{N}(\omega\tau'-\tau\omega')} H_{\tau',\omega'}\circ H_{\tau,\omega},
\end{align}
for every $(\tau,\omega), (\tau',\omega') \in \Z^2_N$.
\begin{restatable}[Existence of Chirps]{theorem}{thmchirps}\label{thm:chirps} 
Let $L\subset \Z_N^2$ be a line. Then,
\begin{enumerate}
    \item The operators $\set{H_{\tau,\omega}}:{(\tau,\omega)\in L}$ have $N$ distinct eigenvalues, and share an orthonormal basis of eigenvectors $\mathcal{B}_L$.
    \item For $S_L \in \mathcal{B}_L$ 
$$
	|\<H_{\tau,\omega}S_L, S_L \>| = \begin{cases}
							1 & \text{ if } (\tau,\omega)\in L \\
							0 & \text{ otherwise }
						   \end{cases}
$$
\end{enumerate}
\end{restatable}
\ 

In fact, the commuting operators from which chirps arise can be seen to come from commutative subgroups of the \textit{Heisenberg-Weyl} group $G_{HW}$ \cite{howe2005nice, howard2006finite, fish2012delay, fish2013delay, fish2013incidence,fish2014}. We can then  utilize the structure of $G_{HW}$ in order to explicitly construct and produce formulas for these functions. 

For the remainder of this section, we assume that $N$ is odd and `` $2^{-1}$" will denote the multiplicative inverse of $2$ modulo $N$, namely $(N+1)/2$.

\begin{restatable}[Formulas for Chirps]{lemma}{thmchirpformula}\label{thm:chirpformula}
We have the following formulas for the orthonormal bases $\mathcal{B}_L$ associated with the line $L\subset \Z_N^2$.
\begin{enumerate}
\item For $a\in\Z_N$ and operators $\set{H_{\tau,a\tau}}:{\tau\in\Z_N}$ corresponding to the line \\ $L = \set{(\tau,a\tau)}:{\tau\in\Z_N}$, we have the orthonormal basis of eigenvectors:
\begin{align}\label{eq:chirp}
\mathcal{B}_L &= \set{S_L^\omega[\tau'] = \frac{1}{\sqrt{N}} \ e^{\frac{2\pi i}{N}\p{2^{-1}a\tau'^2+\omega\tau'}}}:{\omega \in \Z_N}.
\end{align} 
\item For the operators $\set{H_{0,w}}:{\omega\in\Z_N}$ corresponding to the line  $\mathcal{W} = \set{(0,\omega)}:{\omega\in\Z_N}$, we have the orthonormal basis of eigenvectors:
\begin{align}\label{eq:chirpinf}
\mathcal{B}_\mathcal{W} &= \Set{S_\mathcal{W}^\tau[n] = \delta_\tau}.
\end{align}  
\end{enumerate}
\end{restatable}

\begin{restatable}[Chirp cross-correlation.]{lemma}{lemcross}\label{lem:cross}
For distinct lines $L\ne M$, and any choice of chirps $S_L\in \mathcal{B}_L$, $S_M\in \mathcal{B}_M$,
\begin{align*}
    |\mathcal{A}(S_L,S_M)| = \frac{1}{\sqrt{N}}.
\end{align*}
\end{restatable}

\section{Reduction to Sparse FFT} 
\begin{restatable}{theorem}{thmredtosfft}\label{thm:redtosfft}
Given $S\in L^2(\Z_N)$, lines $L$ and $M$:
\begin{itemize}
    \item $L = \set{(\tau,a_1\tau)}:{\tau\in \Z_N}$.
    \item $M = \set{(\tau,a_2\tau)}:{\tau\in \Z_N}$.
\end{itemize}
the values of the ambiguity function of $S$ against the chirp $S^b_L$ on the shifted line $M' = M + (0,\omega)$, $\omega \in Z_N$ are given by the following Fourier coefficients,
\begin{align*}
    \mathcal{A}(S_{L}^b,S)[(\tau,a_2\tau)+(0,\omega)] &= |\mathcal{F}\p{ \ \overline{S_L} \cdot S \  }[(a_1-a_2)\tau+b-\omega]|
\end{align*}
where, $S_L$ is the chirp $S_L = S_L^0$.
\end{restatable}

\section{Discrete Filter Functions}

\begin{definition}[Support] The support of a function complex-valued function $F$ is the subset of its domain on which $F\ne 0$.
\end{definition}

\begin{definition}[Essential Support]\label{defn:esssup} We say that $F\in L^2(\Z_N)$ is essentially supported on an interval $I\subset \Z_N$ if \begin{align*}
    \sum_{\tau\notin I}\big|F[\tau]\big| = O(1/N).
\end{align*}
\end{definition}

\begin{definition}[$(k,k',\delta)$-family of filters]\label{defn:filter} We will refer to $\Set{F_N\in L^2(\Z_N)}_N$ as a $(k,k',\delta)$-family of filters if 
\begin{enumerate}
    \item $F_N$ is essentially supported on the interval $[-k\log N, k\log N]$,
    \item $\mathcal{F}F_N$ is essentially supported on the interval $[-k',k']$,
    \item And lastly,
    \begin{align*}
        \delta\le& \mathcal{F}F_N[\omega] \le 1,
    \end{align*}
    for $\omega$ between $\pm k'\sqrt{\log(1/\delta)}/\sqrt{\log N}$
\end{enumerate}
\end{definition}
\begin{restatable}[Discrete Gaussians]{theorem}{thmgaussian}\label{thm:gaussian} 
We have the following results about the discretization of continous gaussians and their fourier transforms.
\begin{enumerate}
    \item Let $G\in L^2(\Z_N)$ be defined as follows,
    \begin{align}\label{eq:gaussian}
        G[\tau] &= \sum_{m\in \Z}e^{-\pi\p{\frac{\tau}{\sqrt{N}}+m\sqrt{N}}^2}.
    \end{align}
    Then, $G$ is an eigenvector of the discrete fourier transform
    \item Moreover, for any $\sigma\in \R$, if,
        \begin{align}\label{eq:gaussianscale}
            G^\sigma[\tau] &= \sum_{m\in \Z}e^{-\pi\p{\frac{\sigma \cdot \tau}{\sqrt{N}}+m\sqrt{N}}^2}.
        \end{align}
        then,
        \begin{align*}
            \mathcal{F}(G^\sigma) & = \sum_{m\in \Z}e^{-\pi\p{\frac{  \tau}{\sigma \sqrt{N}}+m\sqrt{N}}^2}.
        \end{align*}
\end{enumerate}
\end{restatable}
Note that the only significant term in the expression in \cref{eq:gaussian,eq:gaussianscale} corresponds to $m=0$. We then have the following corollary.
\begin{restatable}[Discrete Gaussian Filters]{corollary}{corfilters}\label{thm:filters}
There exists a $(k\log N,N/k,\delta)$-family of filters for every $0\le \delta < 1$, namely
\begin{align*}
    F_N[\tau]&= e^{-\pi\p{ \frac{\tau^2}{k^2\log N}}}.
\end{align*}
\end{restatable}

\chapter{Complexity bounds and Guarantees for Sparse Channel Estimation (SCE).}
\section{Case 1: $1$-sparse algorithm (``Bit-by-bit")}
\label{app:1sparse}
The algorithm described below can be immediately generalized to the case of $N$ a power of a small prime, but for readability we assume $N$ is a power of $2$. The main idea is based on the following observation.  
\begin{observation}
For $N$ a power of $2$, given $S\in L^2(\Z_N)$ such that
\begin{align}\label{model:k1nonoise}
    S[\tau]&= \sqrt{N} \cdot \alpha_{0} e_{\omega_0}[\tau]  
\end{align}
and if the smallest $k-1$ ``bits" of $\omega_0$ are zero then,
\begin{align*}
 \text{1. the $k^{th}$ bit is zero iff }&\begin{cases}
 |S[r]-S[r+N/2^k]|&=0 \\ 
 |S[r]+S[r+N/2^k]|&=|\alpha_0|
 \end{cases} \\
 \text{2. the $k^{th}$ bit is one iff }&\begin{cases}
 |S[r]-S[r+N/2^k]|&=|\alpha_0| \\ 
 |S[r]+S[r+N/2^k]|&=0
 \end{cases} 
\end{align*}

\end{observation} \ \\
The estimation algorithm then proceeds as follows. We start with an estimate for $\omega_0$ that is simply $\omega=0$. Then, for some fixed $r_1\in \Z_N$ and denoting $A = S[r_1]$, $B= S[r_m+N/2^m]$, if $|A+B| = 0$ we update $\omega=1$; we have just computed the first \textit{bit} of $\omega_0$. \\
At step $m$, where $m$ is between $1$ and $\log_2 N$, we fix $r_m\in \Z_N$ and denote, 
\begin{align*}
    A &= (S\cdot e_{-\omega})[r_m]\\
    B &= (S\cdot e_{-\omega})[r_m+N/2^m]
\end{align*}
if $|A+B| = 0$ we add $2^{m-1}$ to the estimate $\omega$ and we have now computed the $m^{th}$ bit of $\omega_0$. As a result, the estimation algorithm that we eventually produce will be called \textit{Bit-by-bit}. 
Inductively, after $\log_2N$ steps we will have the the estimate $\omega = \omega_0$. \\
More interestingly however, we would like to estimate $\omega_0$ in the presence of noise, namely,  
\begin{align}\label{model:k1noise}
    S[\tau]&= \sqrt{N} \cdot \alpha_{0} e_{\omega_0}[\tau] + \nu_\tau 
\end{align}
where $\nu_n \sim \mathcal{N}(0,\sigma^2)$ i.i.d. and, without loss of generality, $\alpha_0 = 0$ or $1$. \\
The above procedure is then robustified by including certain \textit{averages}. The algorithm will also have a thresholding parameter $\mu \in (0,1)$, to distinguish between the case when $\alpha_0=1$ and $\alpha_0=0$. \\ We present the pseudocode for $\text{Bit-by-bit}_\mu$ below.
\begin{algorithm}[H]
\caption{$\text{Bit-by-bit}_\mu$}\label{alg:k1noise}
\begin{algorithmic}[1]
\REQUIRE $S, N$.
\STATE  $\omega=0$.
\STATE Pick $\set{r_{m,i}\in \Z_N}:{ m = 1, \dots, \log_2(N), i= 1, \dots, T_{bit}}$ uniformly at random.
\STATE For $m = 1, \dots, \log_2(N)$: \\ 
\hspace{.42in} For $i = 1, \dots, T_{bit}$:\\
\hspace{.84in} $A_{m,i} = S[r_{m,i}]\cdot e_{-\omega}[r_{m,i}]$ \\
\hspace{.84in} $B_{m,i} = S[r_{m,i}+N/2^m]\cdot e_{-\omega}[r_{m,i}+N/2^m]$ \\
\hspace{.84in} If $\big|\sum_{i = 1}^{T_{bit}}A_{m,i}+B_{m,i} \big|<\big|\sum_{i = 1}^{T_{bit}} A_{m,i}-B_{m,i}\big|$:\\
\hspace{1.26in} $\omega = \omega+2^{m-1}$\\
\STATE Estimate $\alpha_{0}$: \begin{align*}\alpha= \frac{1}{2\log_2N\cdot T_{bit}}\cdot\p{\sum_{m,i}(A_{m,i}+B_{m,i})}
\end{align*}\\
\ENSURE If $|\alpha|>\mu$ return $\omega$.
\end{algorithmic}
\end{algorithm} \ 
By Lemma \ref{lem:bitbybit}, it suffices to let the parameter $T_{bit} = O(\ln\ln N)$.

\begin{restatable}
{lemma}{lembitbybit}\label{lem:bitbybit} Assume the model \ref{model:k1noise}, and denote $SNR = |\alpha_{\omega}|^2/\sigma^2$. Then there is an implementation of $\text{Bit-by-bit}_\mu$ that,
\begin{enumerate}
    \item Returns $\omega_0$ with probability $1-\delta$,
    \item And exhibits,
    \begin{align*}
    \text{Sample complexity} & = 8\log(\log(N)/\delta)\cdot (SNR)^{-1},\\
    \text{Arithmetic complexity}& \le c \cdot \text{Sample complexity},
    \end{align*}
    for $N$ sufficiently large and $c$ a constant independent of $N, SNR$ and $\delta$.
\end{enumerate}
\end{restatable}

\begin{restatable}[Thresholding]{lemma}{lemthresh}\label{lem:thresh} 
    Assume the model \ref{model:k1noise}, and denote $SNR = |\alpha_{\omega}|^2/\sigma^2$. Fix $\omega \in \Z_N$. For $\tau_1, \dots , \tau_m \in \Z_N$ picked uniformly at random, let 
    \begin{align*}
        \widehat{\alpha} &= \frac{N}{m}\sum_{i=1}^m S[\tau_i]  \cdot e_{-\omega}[\tau_i].
    \end{align*}
    Then, for $\omega = \omega_0$,
        \begin{align*}
            P(||\widehat{\alpha}|-|\alpha_0||>\mu|\alpha_0|) \le \delta
        \end{align*}
    and otherwise,
        \begin{align*}
            P(|\widehat{\alpha}|>\mu|\alpha_0|) \le \delta,
        \end{align*}
    where, 
    \begin{align*}
         \delta & = 2 \exp \p{-m\cdot\frac{\mu^2}{8+2SNR^{-1}}}.
    \end{align*}
\end{restatable}

\section{Case 2: $k>1$, $\omega_j$ uniformly spread}
\begin{restatable}[Filtering]{lemma}{lemfiltering}\label{lem:filtering} Let $m$ be a positive integer such that $m$ divides $N$. For $j = 1, \dots m$, and functions $S$ and $F\in L^2(\Z_N)$  denote,
    \begin{align*}
        S_j &= S * (F \cdot e_{\frac{N}{m}j}).
    \end{align*}
\begin{enumerate} 
    \item If $F\in L^2(\Z_N)$ has  support of size $k$, then for any $\tau \in\Z_N$, we can compute $S_1[\tau], \dots, S_m[\tau]$ simultaneously, using
    \begin{enumerate}
        \item $k$ samples of $S$,
        \item Fewer than $k+m\log m$ operations.
    \end{enumerate}
    \item If $F\in L^2(\Z_N)$ has essential support of size $k$ and $S$ has unit norm, then for any $\tau \in\Z_N$, we can compute estimates $\widetilde{S}_1[\tau], \dots, \widetilde{S}_m[\tau]$ simultaneously, using
    \begin{enumerate}
        \item $k$ samples of $S$,
        \item Fewer than $k+m\log m$ operations,
    \end{enumerate}
    such that $|S_j[\tau] - \widetilde{S}_j[\tau]| = O(1/N)$, for $j = 1, \dots m$.
\end{enumerate}
\end{restatable}
\begin{restatable}[Filtered Noise]{lemma}{lemfilnoise}\label{lem:filnoise} Let $R\in L^2(\Z_N)$ be given by,
\begin{align*}
    R[\tau] = S[\tau] +\nu_\tau,
\end{align*}
where $\nu_\tau \sim\mathcal{N}(0,\sigma^2/N)$ i.i.d. Then, for a discrete gaussian filter $F[\tau] = e^{-\pi\p{ \frac{\tau^2}{4k^2\log N}}}$,
\begin{align*}
    (R*F)[\tau] = (S*F)[\tau] + \widetilde{\nu}_\tau,
\end{align*}
where, $\widetilde{\nu}_\tau \sim \mathcal{N}(0,\widetilde{\sigma}^2/N)$ and,
\begin{align*}
    \widetilde{\sigma}^2 \le \sigma^2/k.
\end{align*}
\end{restatable}

\section{General Case}
\begin{restatable}{lemma}{lemisolate}\label{lem:isolate}
Given 
$$\Set{\omega_1, \dots, \omega_k}\subset \Z_N$$
for $\sigma$ chosen uniformly at random from $ \Z_N^*$,  $\sigma \omega_1$ is isolated with probability at most $2Ck/N$.
\end{restatable}

\begin{restatable}{corollary}{corisolate}\label{cor:isolate}
If $C\le N/8k$, with $O(\ln  k/\delta)$ choices of $\sigma$ uniformly at random from $ \Z_N^*$, for all $i=1, \dots ,k$, $\sigma \omega_i$ is isolated at least once with probability $1-\delta$.
\end{restatable}

\chapter{Background.}

\section{Fourier Transforms}\label{app:ft}

An element $S \in L^2(\Z_N)$ can be represented it in terms of its samples or function values as $S[\tau], \tau \in \Z_N$. This corresponds to representing $S$ in the orthonormal basis of $\delta$-functions, $\set{\delta_\tau}:{\tau\in\Z_N}$. The evaluations $S[\tau]$ may also be referred to as the coordinates of $S$ in \textit{time}. 

We can also express $S$ in the orthonormal basis of exponential functions, $\set{e_\omega}:{\omega\in \Z_N}$. Then, the evaluations $\< S, e_\omega\>$ may be referred to as the coordinates of $S$ in \textit{frequency}. The linear transformation which maps the coordinates of $S$ in time to its coordinates in frequency is known as the \textit{discrete fourier transform}, or DFT for short, and we will denote it by $\mathcal{F}$,
\begin{align*}
\mathcal{F}S[\omega] &= \< S, e_\omega\> & \omega \in \Z_N.
\end{align*}
In terms of the samples $S[\tau]$, $\mathcal{F}S$ can be expressed as,
\begin{align*}
\mathcal{F}S[\omega]&=\sum_{\tau\in \Z_N} S[\tau] \cdot \overline{e_{\omega}[\tau]} & \omega \in \Z_N.
\end{align*}
\
If the coordinates of $S$ are given to us in time, we could naively compute the discrete fourier transform of $S$ in $O(N \cdot \text{supp}(S))$ operations, where $\text{supp}(S)$ denotes the size of the support of $S$ in time. However, with the \textit{Fast Fourier Transform} (FFT) algorithm we can compute $\mathcal{F}(S)$ in $O(N\log N)$ operations even if $S$ on all of $\Z_N$.

Analogously, the \textit{continuous fourier transform} of a function $s\in L^2(\R)$, is a complex valued function on the real line which we will denote\footnote{The same symbol $\mathcal{F}$ is used for both the continuous and discrete fourier transforms; which one is meant will be evident from the context.} by $\mathcal{F}s$. This function is given by,
\begin{align*}
    \mathcal{F}s(f) &= \int_\R s(t)e^{-2\pi i ft} \ dt & f\in\R
\end{align*}


\section{Lines in $\Z_N$.} \label{app:lines}

As discussed, a line in $\Z_N^2$ is defined similar to lines in the real plane $\R^2$ -- a line in $\R^2$ is the $\R$-linear span of some point in the plane, we use ``line" to refer to what may specifically called ``lines through the origin" in other contexts.  

However, consider the following example.
\begin{example}
Consider $N = 4$. The points $(1,2),(2,2)$ lie in the $Z_4^2$ plane. The $\Z_4$-span of $(1,2)$ is given by:
\begin{align*}
\set{(n,2n)}:{n \in \Z_4} & = \Set{(0,0),(1,2),(2,0), (3,2)}
\end{align*}
\begin{align*}
\set{(2n,2n)}:{n \in \Z_4} & = \Set{(0,0), (2,2)}
\end{align*}
The $Z_4$-span of certain points may consist of fewer than $N=4$ elements. 
\end{example}

For our purposes, we will use the term \textit{generic line} to refer to a set of $N$ points which is the $\Z_N$-span of a point. For details on more general sets that could be used see Sec \ref{app:lines}

\begin{definition}\label{def:genline} Generic line: A generic line $L$ is the $\Z_N$-span of a point $(a,b) \in \Z_N^2$, which consists of $N$ points. We say $L$ is generated by $(a,b)$. 
\end{definition}

\begin{lemma}\label{lem:lines} The $\Z_N$-span of a point $(a,b) \in \Z_N^2$ consists of $N/\gcd (a,b,N)$ points.
\end{lemma}
\begin{proof} The $\Z_N$-span of $(a,b)$ is a group under addition. For $N' = N/\gcd (a,b,N)$, one can verify that the span is isomorphic to $\Z_N'$. 
\end{proof}
Therefore, we have the following corollary.
\begin{corollary}
A generic line $L$ is the $\Z_N$-span of a point $(a,b) \in \Z_N^2$, where $\gcd(a,b,N) =1$.
\end{corollary}

From now forth, for brevity we will use the terms \textit{generic line} and \textit{line} interchangeably. 

In the real plane, distinct lines always intersect at exactly one point, namely the origin. Consider the following example in $\Z_4$.

\begin{example} Let $L$ be the line generated by $(1,2)$ and $M$ the line generated by $(1,0)$ in $\Z_4^2$. A point lies in the intersection of $L$ and $M$ is there is some $\ell,\mu \in \Z_4$ such that:
\begin{align*}
& (\ell, 2\ell) = (\mu,0) \\
\text{i.e. }& \ell = \mu \text{ and  } 2\ell = 0  
\end{align*}
So $L$ and $M$ intersect at two points $(0,0)$ and $(2,0)$.
\end{example}

In order to identify the time-frequency shift, we need the lines to have a unique point of intersection. We have the following definition.

\begin{definition} A pair of lines $L$ and $M$ in $Z_N^2$ are said to be transversal if they intersect only at the origin.
\end{definition}

We use the following lemma and corollary.

\begin{lemma}\label{lem:linesint} If two lines $L$ and $M$ in $\Z_N^2$ are generated by $(a_1,b_1)$ and $(a_2,b_2)$ respectively and
\begin{align*}
\gcd \p{\begin{vmatrix}
a_1 & a_2 \\
b_1 & b_2
\end{vmatrix},N}
\end{align*}
then $L$ and $M$ are transversal.
\end{lemma}
\begin{proof} 
For each point of intersection of $L$ and $M$ there is some $l,m$ in $\Z_N$ such that 
\begin{align*}
    la_1 - ma_2 &= 0 \\
    lb_1 - mb_2 &= 0
\end{align*}
This has a unique solution if $\begin{vmatrix}
a_1 & a_2 \\
b_1 & b_2
\end{vmatrix}$ is invertible.
\end{proof}

\begin{restatable}{lemma}{genline}\label{lem:genline}
For $k$ points $\Set{p_1, \dots , p_k}$ in $\Z_N^2$, and three distinct pairwise transversal lines $L,M,K$ with slopes chosen uniformly at random, the triple intersections of the shifted lines 
\begin{align*}
    p_i+L, \  p_i+M, \ p_i+K && i= 1, \dots, k
\end{align*}
uniquely identify the points with probability at least $1 - \frac{k^3}{N-2}$.
\end{restatable}

\chapter{Proofs.}\label{appD}

\section{Justifications behind the Model.}
\subsection{Time scale to frequency shift.}
\narrowband*
\begin{proof} Note that by the definition, 
$$s(t) = e^{2\pi if_ct} s_0(t)$$
where, $f_c$ is the carrier frequency and $s_0 \in L^2(\R)$ (often known as the \textit{baseband signal}) has fourier support in the interval $[-W/2,W/2]$.

One can check that it suffices to show, 
$$|s_0(a_0t) - s_0(t)| \le O\p{f_0t\cdot \frac{W}{f_c}}$$

By Bernstein's Inequality (Theorem 2.3.17, \cite{pinsky2008introduction})
\begin{align}\label{eq:bi}
|s'_0(t)| \le 2\pi W \sup_{x\in \R}|s_0(t)| 
\end{align}

By the Mean value theorem and \cref{eq:bi},
\begin{align*}
|s_0(a_0t) - s_0(t)| &\le const \cdot |(a_0-1)t|W \\
& = const \cdot |f_c(a_0-1)t|\cdot \frac{W}{f_c} \\
& = const \cdot |f_0t|\cdot \frac{W}{f_c} 
\end{align*}

\end{proof}

\section{Moving between Digital and Analog settings.}
\propadda*
\begin{proof} \ \\
Let $S\in L^2(\Z_N)$. Then, by formula \ref{eq:daformula},
\begin{align*}
    \{h_{t_0,f_0}\circ \text{D-to-A}(S)\}(t) &= e^{2\pi if_0t}\cdot e^{2\pi if_c(t-t_0)} \cdot \p{\sum_{\tau\in \Z_N} S[\tau]\cdot \text{sinc}_W\p{t-\frac{\tau}{W}-t_0}}
\end{align*}
Consider the natural map, 
\begin{align*}
    &\Lambda_{T,W} = \frac{1}{W}\Z\times \frac{1}{T} \Z \ \to \ \Z\times \Z 
\end{align*}
If $(\overline{\tau}_0,\overline{\omega}_0) \in \Z^2$ denotes the image of $(t_0,f_0)$ under this map then,
\begin{align*}
    \{h_{t_0,f_0}\circ \text{D-to-A}(S)\}(t) &= e^{2\pi i\p{\frac{\overline{\omega}_0}{T}}t}\cdot e^{2\pi i f_c\p{t-\frac{\overline{\tau}_0}{W}}} \cdot \p{\sum_{\tau\in \Z_N} S[\tau]\cdot \text{sinc}_W\p{t-\frac{\tau+\overline{\tau}_0}{W}}}    
\end{align*}
If we denote $R = \text{A-to-D} \circ h_{t_0,f_0}\circ \text{D-to-A}(S)$, by formula \ref{eq:adformula},
\begin{align}\label{eq:1311}
   R[\tau] &= \sum_{m\in \Z} \{h_{t_0,f_0}\circ \text{D-to-A}(S)\} \p{\frac{\tau}{W}+mT}. 
\end{align}
Now consider another natural map, 
\begin{align*}
    \Z\times \Z \ \to \ \Z_N\times \Z_N
\end{align*}
Let $(\tau_0,\omega_0) \in \Z_N^2$ denote the image of $(\overline{\tau}_0,\overline{\omega}_0)$ under the above map. Then,
\begin{align}\label{eq:1312}
    \sum_{m\in \Z} \{h_{t_0,f_0}\circ \text{D-to-A}(S)\} \p{\frac{\tau}{W}+mT} &= e^{2\pi i \frac{\omega_0\tau}{N}} S[\tau - \tau_0]
\end{align}
where equality follows from the fact that:
\begin{enumerate}
    \item $N=TW$, and so $e^{2\pi i\p{\frac{\overline{\omega}_0\tau}{TW}}} = e^{2\pi i\frac{\omega_0\tau}{N}}$,
    \item $f_c$ is a multiple of $W$, and so $ e^{2\pi i f_c\p{\frac{\tau}{W}-\frac{\overline{\tau}_0}{W}}} = 1$, 
    \item and lastly,
    \begin{align*}
        \sum_{\tau'\in \Z_N} S[\tau]\cdot \text{sinc}_W\p{\frac{\tau}{W}-\frac{\tau'+\overline{\tau}_0}{W}}&= S[\tau - \tau_0]
    \end{align*}
    since, 
   \begin{align*}
    \text{sinc}_W\p{\frac{\tau}{W}} = \begin{cases}
   1 &\text{ if } \tau =0 \\
   0 &\text{ otherwise}
   \end{cases}
   \end{align*}
\end{enumerate}
Following along the equalities in \cref{eq:1311} and \cref{eq:1312}, we see that for any $S\in L^2(\Z_N)$,
\begin{align*}
    \{\text{A-to-D} \circ h_{t_0,f_0}\circ \text{D-to-A}(S)\}[\tau] & = e^{2\pi i \frac{\omega_0\tau}{N}} S[\tau - \tau_0] \\
    &= \{H_{\tau_0,\omega_0}S\}[\tau].
\end{align*}
\end{proof}

\section{Underlying Algebraic Structure.}
\subsection{Construction of Chirps.}

\thmchirps*
\begin{proof}\ \\
    Fix a line $L$. Let $H_L=H_{\tau,\omega}$ for some $(\tau,\omega)\in L$. Fix another line $M\ne L$ and a generator $(\tau',\omega')\in M$; let $H_M = H_{\tau',\omega'}$. \\
    We use the fact that if $v$ is an eigenvector of $H_L$, then $H_Mv$ is also an eigenvector of $H_L$ and, moreover, $H_Mv\ne v$. This statement is verified below:\\
    It follows from the Heisenberg commutation relation \ref{eq:commutation} that,
    \begin{align*}
        H_L(H_Mv) &= (H_L\circ H_M)v\\
                  &= e^{\frac{2\pi i}{N}(\omega\tau'-\tau\omega')}(H_M\circ H_L)v \\
                  &= e^{\frac{2\pi i}{N}(\omega\tau'-\tau\omega')} \cdot H_M ( H_Lv) \\
                  &=  e^{\frac{2\pi i}{N}(\omega\tau'-\tau\omega')} \cdot \lambda \cdot H_M v.
    \end{align*}
    Therefore, $H_Mv$ is also an eigenvector of $H_L$. Denoting $e^{\frac{2\pi i}{N}(\omega\tau'-\tau\omega')}$ by $\eta$, since $M\ne L$, we have that $\eta\ne 1$ and so these eigenvalues -- and the corresponding eigenvectors -- are distinct.
    \begin{enumerate}
        \item 
        We now show that the operators $\set{H_{\tau,\omega}}:{(\tau,\omega) \in L}$, each have $N$ distinct eigenvalues:\\
        Every operator on a finite, positive dimensional complex vector space has an eigenvector, so let $v$ denote an eigenvector of $H_L$ of eigenvalue $\lambda$.\\
        By a similar argument as the above, we can show that $v, H_Mv, H_M^2v, \dots, H_M^{N-1}v$ are eigenvectors with distinct eigenvalues $\lambda\eta^2, \cdots, \lambda\eta^{N-1}$, respectively. \\
        Next, we have that by the commutation relation \ref{eq:commutation}, the operators $\set{H_{\tau,\omega}}:{(\tau,\omega) \in L}$, commute (since the expression $\omega\tau'-\tau\omega'$ vanishes for $(\tau,\omega)$, $(\tau',\omega)$ on the same line). \\
        It then follows that they share an orthonormal eigenbasis.
        \item Fix $S_L \in \mathcal{B}_L$. Let $(\tau,\omega)\in L$. \\ Then, $H_{\tau,\omega}S_L =\lambda S_L$ for some $\lambda \in \C$. Since $(H_{\tau,\omega})^N$ is the identity operator, we must have $|\lambda|=1$. Therefore,
        \begin{align*}
            |\<S_L,H_{\tau,\omega}S_L\>| & = |\lambda||\<S_L,S_L\>|= 1
        \end{align*}
        For $(\tau,\omega)\notin L$, then by a similar argument as above, $H_L(H_{\tau,\omega}S_L) = \lambda\eta H_{\tau,\omega}S_L$, for some $\eta \ne 1$. So it follows that,
        \begin{align*}
            |\<S_L,H_{\tau,\omega}S_L\>| = 0.
        \end{align*}
    \end{enumerate}
\end{proof}
The chirp signals were realized as elements of eigenbases corresponding to certain collections of commuting operators. These operators can be seen to come from certain commuting subgroups of the Heisenberg-Weyl group $G_{HW}$ \cite{howe2005nice, howard2006finite,fish2012delay,fish2013delay,fish2013incidence,fish2014performance,fish2014}, in particular, via a \textit{group representation} of $G_{HW}$. 

While we do not go into the details of this representation, we will adopt some notation that is ``inspired" by it, in order to conveniently produce formulas for chirps.\\
Note that $N$ is assumed to be odd throughout the remainder of this section, and $2^{-1}$ denotes the multiplicative inverse of ``$2$" modulo $N$, i.e. $2^{-1} = (N+1)/2$.
\begin{definition}[Heisenberg Operators] For $(\tau,\omega)\in \Z_N^2$, we define the Heisenberg operator $\pi(\tau,\omega)$ as,
\begin{align*}
    \pi(\tau,\omega) &= e^{\frac{2\pi i}{N} (-2^{-1}\tau \omega)} H_{\tau,\omega}.
\end{align*}
\end{definition}
We would also like to introduce \textit{chirp operators} on $\Z_N$ and since they are members of a larger collection, namely the \textit{Weil operators} on $L^2(\Z_N)$, we first define those. Moreover, in describing Weil operators, we will denote the \textit{special linear group} of order $2$ over $\Z_N$ as $SL_2(\Z_N)$,
\begin{align*}
    SL_2(\Z_N)= \set{\renewcommand\arraystretch{.5}\begin{pmatrix}a & b \\ c & d\end{pmatrix}}:{a,b,c,d \in \Z_N \text{ and } ad-bc = 1}
\end{align*}
\begin{definition}[Weil Operators]\label{def:weil} These are the unique collection of operators \\ $\set{\rho(g)}:{g\in SL_2(\Z_N)}$ such that,
\begin{enumerate}
    \item $\rho(gh) = \rho(g) \circ \rho(h)$
    \item $\rho(g)\circ \pi(\tau,\omega)\circ \rho(g)^{-1} = \pi(g(\tau,\omega))$
\end{enumerate}
\end{definition}
A justification for the existence and uniqueness of such operators can be found in Section II.A of \cite{fish2014}. \\
We can now define the following.
\begin{definition}[Chirp Operators] For $a \in \Z_N$, we define the chirp operator $\rho_a$ as,
\begin{align*}
    \rho_a &= \rho\renewcommand\arraystretch{.5}\begin{pmatrix}1 & 0 \\ a & 1\end{pmatrix}
\end{align*}
\end{definition}
One can check \cite{fish2014} that, by definition,
\begin{align}\label{eq:chirpopformula}
    \p{\rho_a S} [\tau] = e^{\frac{2\pi i}{N}(2^{-1}a\tau^2)} S[\tau]
\end{align}

\thmchirpformula*
\begin{proof} \ \\
    When we think of elements in $L^2(\Z_N)$, the first ones that we might write down maybe $\delta$-functions, namely,
    \begin{align*}
        \delta_\tau [\tau'] 
        =\begin{cases}
            1 & \text{ if }\tau = \tau' \\
            0 & \text{ otherwise.}
        \end{cases}
    \end{align*}
    It's a quick check that $\delta_\tau$ are eigenvectors of the frequency shift operators, $\set{H_{0,\omega}}:{\omega\in \Z_N}$, with eigenvalue $e^{\frac{2\pi i}{N} \omega\tau}$. \\
    Perhaps, after a little bit more thought, we might recall the discrete fourier transform  and write down the complex exponentials,
    \begin{align*}
        e_{\omega}[\tau'] = e^{\frac{2\pi i}{N}\omega \tau'}
    \end{align*}
    It is then a quick check that $e_\omega$ are eigenvectors of the time shift operators, $\set{H_{\tau,0}}:{\tau\in \Z_N}$, with eigenvalue $e^{-\frac{2\pi i}{N} \omega\tau}$. \\
    For $a\in \Z_N$, by \cref{def:weil}, \begin{align}\label{eq:weil}
        \pi\p{\tau,a\tau} &= \rho_a\circ  \pi\p{\tau,0}\circ \rho_a^{-1}.
    \end{align}
    Using \cref{eq:weil} and the fact that $e_\omega$ is an eigenvector of $H_{\tau,0}= \pi(\tau,0)$, we have,
    \begin{align*}
        \pi\p{\tau,a\tau}\p{\rho_a e_\omega} & = \rho_a\circ  \pi\p{\tau,0}\circ \rho_a^{-1}\p{\rho_a e_\omega} \\
        & = \rho_a\p{  \pi\p{\tau,0}e_\omega} \\
        &= e^{-\frac{2\pi i}{N} \omega\tau}\cdot \rho_a e_\omega
    \end{align*}
    It follows that $\rho_ae_\omega$ is an eigenvector of $H_{\tau,a\tau}$ (in fact, with eigenvalue $e^{\frac{2\pi i}{N}(2^{-1}a\tau^2-\omega\tau)}$).
\end{proof}

\lemcross*
\begin{proof} \ \\
    We first show this is true for $L=\mathcal{W}$ and $M=\mathcal{T}$. Then,
    $S_L = \delta_\tau$ for some $\tau \in\Z_N$ and $S_M = e_\omega$ for some $\omega\in \Z_N$ and,
    \begin{align}\label{eq:cross}
        |\<\delta_\tau, e_\omega\>| = 1/\sqrt{N}
    \end{align}
    Moreover, since \cref{eq:cross} is true for any $\tau$ and $\omega$ in $\Z_N$ we have,
    \begin{align*}
        \mathcal{A}(\delta_\tau,e_\omega)[\tau', \omega'] & = |\<\delta_\tau,H_{\tau',\omega'}e_\omega\>|\\
        &= |\<\delta_\tau, e_{\omega+\omega'}\>|\\
        &= 1/\sqrt{N}.
    \end{align*}
    A similar argument will show that $|\mathcal{A}(\delta_\tau,S_M)| = |\mathcal{A}(S_M,\delta_\tau)|= 1/\sqrt{N}$, for any $\tau \in \Z_N$ and $S_M \in \mathcal{B}_M$, $M\ne \mathcal{W}$.
    \\
    Now, more generally, let $L = (\tau,a\tau)$, for $a\in \Z_N$, let $g \in SL_2(\Z_N)$ be an element that fixes $M$ and maps,
    \begin{align*}
        g: (1,a)\mapsto (0,1).
    \end{align*}
    The Weil operator $\rho_g$ then maps $S_L \mapsto \delta_\tau$ for some $\tau \in \Z_N$. In other words, $\rho_g S_L\in  \mathcal{B}_{\mathcal{W}}$ as demonstrated below.
    \begin{align*}
        H_{0,\omega}(\rho_g S_L) & =  \rho_g H_{(g^{-1}(0,\omega))}S_L &\text{(definition of Weil operator)}\\
                               & =   \rho_g H_{(\tau,a\tau)}S_L &(\text{for some }\tau\in \Z_N) \\
                               &=  \lambda \cdot \rho_gS_L & (S_L\in \mathcal{B}_L).
    \end{align*}    
    Since the Weil operators are unitary we have,
    \begin{align*}
        |\mathcal{A}(S_L,S_M)| 
        & = |\mathcal{A}(\rho_gS_L,\rho_gS_M)| \\
        & = |\mathcal{A}(\delta_\tau, S_M)|\\
        & = 1/\sqrt{N}.
    \end{align*}
\end{proof}

\subsection{Reduction to SFFT.}
\thmredtosfft*
\begin{proof} \ \\
As given in \ref{eq:chirp},
\begin{align*}
    S^b_L[\tau] &= \frac{1}{\sqrt{N}} \ e^{\frac{2\pi i}{N}\p{2^{-1}a_1\tau^2+b\tau}}
\end{align*}
One can check that,
\begin{align*}
    H_{\tau, a_2\tau}S^b_L
    & = e^{\frac{2\pi i}{N}\p{2^{-1}a_1\tau^2-b\tau}} \cdot \p{S_L \cdot e_{(a_2-a_1)\tau+b}}
\end{align*}
It then immediately follows that,
\begin{align*}
    H_{\tau, a_2\tau+\omega}S^b_L
    & = e^{\frac{2\pi i}{N}\p{2^{-1}a_1\tau^2-b\tau}} \cdot \p{S_L \cdot e_{(a_2-a_1)\tau+b+\omega}}
\end{align*}
We then have,
\begin{align*}
    |\mathcal{A}(S,S^b_{L})[(\tau,a_2\tau+\omega)]| &= |\<S, H_{\tau, a_2\tau+\omega}S^b_L\>|\\
    &= \underbrace{|e^{\frac{2\pi i}{N}\p{2^{-1}a_1\tau^2-b\tau}}|}_{=1}\cdot |\<S,S_L \cdot e_{(a_2-a_1)\tau+b+\omega}\>|\\
    & = |\<S\cdot \overline{S_L} \ , \  e_{(a_2-a_1)\tau+b+\omega}\>|\\
    & = |\mathcal{F}(S\cdot \overline{S_L})[(a_2-a_1)\tau+b+\omega]| &\text{ (by definition)}
\end{align*}
It then follows that,    
\begin{align*}
    |\mathcal{A}(S^b_{L},S)[(\tau,a_2\tau +\omega)]| &= |\<S^b_L, H_{\tau, a_2\tau+\omega}S\>| \\ 
    & = |\<H_{(-\tau, -a_2\tau-\omega)}S^b_L, S\>| & (H_{\tau,\omega}\text{ is unitary.})\\
    & = |\<S,H_{(-\tau, -a_2\tau-\omega)}S^b_L\>| \\
    & = |\mathcal{A}(S,S^b_{L})[(-\tau,-a_2\tau-\omega)]| \\ 
    & = |\mathcal{F}(\overline{S_L}\cdot S )[(a_1-a_2)\tau+b-\omega]|.
\end{align*}

\end{proof}

\subsection{Discrete Filter Functions.}
The proof of \cref{thm:gaussian} is based on the well-known Poisson summation formula, so we would like to present a statement for the same. However, in order to do so we will first introduce some notation. \\
For any positive $\lambda\in\R$, we denote the space of $\lambda$-periodic functions on the real line as $L^2(\R/\lambda\Z)$. We can then define the following operators:
\begin{definition}[Averaging and Evaluation operators] \ 
\begin{enumerate}
    \item For any positive $\lambda\in\R$, we define an averaging operator,
    \begin{align*}
        &\text{Av}_\lambda : L^2(\R) \to L^2(\R/\lambda\Z) \\
        &\p{\text{Av}_{\lambda}s} (t) = \sum_{m\in \Z} s(t+m\lambda).
    \end{align*}
    Analogously, we have an averaging operator, $\text{Av}_{N}:L^2(\Z) \to L^2(\Z_N)$. 
    \item For any positive $\lambda\in\R$, we define an evaluation operator,
    \begin{align*}
    &\text{Ev}_{\lambda} : L^2(\R) \to L^2(\Z) \\
    &\p{\text{Ev}_{\lambda}s} [n] = s\p{\lambda n}
    \end{align*}
    We can restrict the above evaluation operator to $L^2(\R/\lambda\Z) \subset L^2(\R)$ and, in fact, we have,
    \begin{align*}
        &\restriction{\text{Ev}_{\lambda/N}}{L^2(\R/\lambda\Z)} : L^2(\R/\lambda\Z) \to L^2(\Z/N) 
    \end{align*}
    For readability, we will simply denote $\restriction{\text{Ev}_{\lambda/N}}{L^2(\R/T\Z)}$ as $\text{Ev}_{\lambda/N}$.
\end{enumerate}
\end{definition}
\begin{theorem}[Poisson Summation]\label{thm:poisson} The following diagram commutes:
\begin{figure}[H]
    \centering
\begin{tikzpicture}[descr/.style={fill=white}]
\matrix(m)[matrix of math nodes, row sep=3em, column sep=2.8em,
text height=1.5ex, text depth=0.25ex]
{L^2(\mathbb{R})&&&& L^2(\mathbb{R}) \\ &\\ L^2(\R/\lambda\Z)&&&& L^2(\mathbb{\Z})\\&\\ L^2\mathbb{(Z}_N)&&&& L^2(\mathbb{Z}_N)\\};
7
\path[->,font=\scriptsize]
(m-1-1) edge node[right] {$\text{Av}_{\lambda\Z}$} (m-3-1)
(m-1-5) edge node[right] {$\text{Ev}_{1/\lambda}$} (m-3-5)
(m-1-1) edge node[above] {Continuous Fourier transform} (m-1-5)
(m-3-1) edge node[above] {Fourier Series} (m-3-5)
(m-5-1) edge node[above] {Discrete Fourier Transform} (m-5-5)
(m-3-1) edge node[right] {$\text{Ev}_{\lambda/N}$} (m-5-1)
(m-3-5) edge node[right] {$\text{Av}_{N\Z}$} (m-5-5);
\end{tikzpicture}

\end{figure}
\end{theorem}
\thmgaussian*
\begin{proof} \ 
\begin{enumerate}
    \item Consider the continuous Gaussian, which we will denote by $g\in L^2(\R)$,
\begin{align*}
    g(t) = \frac{1}{\sqrt{2\pi}} \ e^{-\pi t^2}
\end{align*}
$g$ is an eigenfunction of the continuous fourier transform,\footnote{The existence of such an eigenfunction follows from the irreducibility of the continuous \textit{Heisenberg-Weyl} representation on $L^2(\R)$.}
so we have,
\begin{align*}
    \mathcal{F}g = g
\end{align*}
It seems natural to utilize this eigenfunction in order to produce an eignevector of the discrete fourier transform. Moreover, recall that Shannon provides a family of maps, which we call ``$\text{A-to-D}$" (\ref{eq:adformula}), that can produce an element of $L^2(\Z_N)$ from a function in $L^2(\R)$. \\
Note that $\text{A-to-D} = \text{Ev}_{T/N}\circ \text{Av}_{T\Z}$. Further, Poisson summation \ref{thm:poisson} tells us that,
\begin{align*}
    \mathcal{F}(\text{Ev}_{T/N}\circ \text{Av}_{T}g)&= \text{Av}_{N}\circ \text{Ev}_{1/T}g.
\end{align*}
If the ``lattices" $\frac{T}{N}\Z$ and $\frac{1}{T}\Z$ are equal, then we have that
\begin{align*}
    \text{Av}_{N}\circ \text{Ev}_{1/T}=\text{Ev}_{T/N}\circ \text{Av}_{T}.
\end{align*}
The above lattices will be equal if $T/N = 1/T$, i.e. $T = \sqrt{N}$. In other words,
\begin{align*}
    \text{Ev}_{1/\sqrt{N}}\circ \text{Av}_{\sqrt{N}}g
\end{align*}
is an eigenvector of the discrete fourier transform.

\item In order to see why the second statement is true, we first note how scaling interacts with the fourier transform. In particular, the following diagram commutes,
\begin{figure}[H]
    \centering
\begin{tikzpicture}[descr/.style={fill=white}]
\matrix(m)[matrix of math nodes, row sep=3em, column sep=2.8em,
text height=1.5ex, text depth=0.25ex]
{L^2(\mathbb{R})&& L^2(\mathbb{R}). \\ };
7
\path[->,font=\scriptsize]
(m-1-1) edge node[above] {$\mathcal{F}$} (m-1-3)
(m-1-1) edge [loop above]  node{$\sigma$} (m-1-1)
(m-1-3) edge [loop above] node {$\sigma^{-1}$} (m-1-3);
\end{tikzpicture}
\end{figure}
In other words,
\begin{align*}
    \mathcal{F}(G^{\sigma}) & = (\mathcal{F}G)^{\sigma^{-1}}. 
\end{align*}
Now, the result follows immediately by the same argument as in part 1.
\end{enumerate}
\end{proof}
\corfilters*

\begin{proof} \ \\
We first examine the essential support of $F_N$:\\
Let $I$ denote the interval $[-k\log N,k\log N]\subset \Z_N$
\begin{align*}
    \sum_{\tau \notin I}|F_N[\tau]| &\le 2e^{-\pi \p{\log(N)}}\cdot \sum_{\tau = 0}^{N/2}  e^{-\pi \p{\frac{\tau^2}{k^2\log N}}} \\
    & \le \frac{1}{N^\pi} \cdot \sum_{\tau = 0}^{N/2} e^{-\p{\frac{\pi }{k^2\log N}}\tau}\\ 
    & = \frac{1}{N^\pi} \cdot \p{\frac{1-e^{- \p{\frac{\pi N}{2k^2\log N}}}}{1-e^{- \p{\frac{\pi}{k^2\log N}}}}} \approx \frac{1}{N^\pi}\p{1+e^{- \p{\frac{\pi}{k^2\log N}}}}
\end{align*}
Therefore, for $N$ sufficiently large.
\begin{align*}
    \sum_{\tau \notin I}|F_N[\tau]|& \le \frac{2}{N^\pi} = O(1/N). 
\end{align*}
Next let's consider the support of $\mathcal{F}F_N$:\\
Let $\widetilde{F}$ be defined as, 
\begin{align*}
    \widetilde{F}_N[\tau] &= e^{-\pi \p{\frac{k^2\log N\tau^2}{N^2}}}
\end{align*}
Then, simply by replacing $k$ in the first argument by $N/(k\log N)$, it follows that $\widetilde{F}_N$ is essentially supported on the interval $[-N/k,N/k]$. In other words,
\begin{align}\label{eq:rando1}
    \sum_{\tau \notin I}|\widetilde{F}_N[\tau]|& = O(1/N). 
\end{align}
\\ 
Moreover,
\begin{align}\label{eq:rando2}
    \widetilde{F}_N\b{\omega} \ge  \delta 
\end{align}
for $\tau$ in $\Z_N$ between $\pm N\sqrt{\log(1/\delta)}/\pi k\sqrt{\log N}$. 
\\
By \cref{thm:gaussian}, for $\sigma = \sqrt{N}/(k\sqrt{\log N})$,
\begin{align*}
    G^\sigma &= \sum_{m\in \Z}e^{-\pi\p{\frac{\tau}{(k\sqrt{\log N})}+m\sqrt{N}}^2}\\
    \mathcal{F}(G^\sigma) & = \sum_{m\in \Z}e^{-\pi\p{\frac{ k\sqrt{\log N} \tau}{N}+m\sqrt{N}}^2}.
\end{align*}
The only significant term in the first sum above is $F_N[\tau]$, corresponding to $m=0$, and the only significant term in the second sum is $\widetilde{F}_N[\tau]$, again corresponding to $m=0$. Now, since \cref{eq:rando1,eq:rando2} hold true for $\widetilde{F}_N$, the same will be true for $\mathcal{F}F_N$ as well and the result follows.
\end{proof}
    
\section{Complexity bounds and guarantees for SCE.}

\subsection{Case 1: $1$-sparse algorithm (``Bit-by-bit")}
\lembitbybit*
\begin{proof} \ \\
First consider the case that $T_{bit} =1$ in \cref{alg:k1noise}. Fix $r_1\in\Z_N$, denote
\begin{align*}
    A_1&= S[r_1]\cdot e_{-\omega}[r_1]
\end{align*}
Fix a positive integer $m$, $1\le m\le \log_2N$, and let
\begin{align*}
    B_1&=S[r_1+N/2^m]\cdot e_{-\omega}[r_1+N/2^m]
\end{align*}
Denote,
\begin{align*}
    \mu^A_1 & = \mu_{r_1} \\
    \mu^B_1 & = \mu_{r_1+N/2^m}\cdot e_{-\omega}[N/2^m]
\end{align*}
If the $m^{th}$ smallest bit of $\omega$ is zero: 
\begin{align*}
    A_1+B_1 &= \alpha_\omega+\mu^A_1+\mu^B_1 \\
    A_1-B_1 &= \mu^A_1-\mu^B_1
\end{align*}
$|A_1+B_1|<|A_1-B_1|$, if $2|\mu^A_1|>|\alpha_\omega|$ or $2|\mu^B_1|>|\alpha_\omega|$. In other words, \cref{alg:k1noise} wrongly estimates this bit if $2|\mu^A_1|>|\alpha_\omega|$ or $2|\mu^B_1|>|\alpha_\omega|$. \\ It is an exercise to check that the same is true of the algorithm, if the $m$ smallest bit of $\omega$ is $1$.
\\
Let's consider $T_{bit}>1$ trials and fix $r_1, \dots r_{T_{bit}} \in \Z_N$. Then, 
\begin{align*}
    P\p{2\cdot \bigg|\frac{1}{T_{bit}}\sum_{i=1}^{T_{bit}}\mu_i\bigg|>|\alpha_\omega|} &\le e^{-\frac{T_{bit}\cdot SNR}{8}}  &\text{(by Hoeffding's inequality)}.
\end{align*}
Therefore if we compare the averages,
\begin{align*}
    \bigg|\sum_{i=1}^{T_{bit}} A_i+B_i\bigg|, \bigg|\sum_{i=1}^{T_{bit}} A_i-B_i\bigg|,
\end{align*}
the probability that \cref{alg:k1noise} wrongly estimates the $m^{th}$ bit is,
\begin{align*}
    P(m^{th}\text{ bit is wrongly estimated}) &\le 2 \cdot e^{-\frac{T_{bit}\cdot SNR}{8}}  &\text{(by a union bound.)}
\end{align*}
For $T_{bit} \ge  8\log(\log(N)/\delta)\cdot (SNR)^{-1}$, the probability the $m^{th}$ bit is wrongly estimated is at most
\begin{align*}
    P(m^{th}\text{ bit is wrongly estimated}) \le \frac{\delta}{\log N}.
\end{align*}
By a union bound again, the probability that any bit is wrongly estimated is at most $\delta$. \\ Since sample and arithmetic complexity are both a constant multiple of $T_{bit}$, the result follows.
\end{proof}

The justification we present for the next lemma will make use of the following definition \cite{wainwright2019high}.
\begin{definition}[Sub-Gaussian random variable] A random variable $X$ with mean $\mu = \mathbb{E}(X)$ is sub-Gaussian if there is a positive number $\sigma$ such that 
\begin{align*}
    \mathbb{E}(e^{\lambda|X-\mu|})\le e^{\sigma^2\lambda^2/2}.
\end{align*}
We then say $X$ is sub-Gaussian with parameter $\sigma^2$.
\end{definition}
We are interested in sub-Gaussian random variables because they have ``small tails", namely, 
\begin{align*}
    P(|X-\mu|> t) \le 2e^{-t^2/2\sigma^2}
\end{align*}
Below we list other facts around sub-Gaussian random variables that will be useful to us. We state them without proof, justifications can be found in Chapter 2 of \cite{wainwright2019high}. \\
\textit{\textbf{Useful facts:}}
\begin{enumerate}
    \item A normal random variable $\nu \sim \mathcal{N}(0,\sigma^2)$ is sub-Gaussian with parameter $\sigma^2$.
    \item A bounded random variable $|X|\le \sigma$ is sub-Gaussian with parameter $4\sigma^2$.
    \item If $X_1$ and $X_2$ are independent sub-Gaussian random variables, with parameters $\sigma_1^2$ and $\sigma_2^2$ respectively, then $X_1+X_2$ are sub-Gaussian with parameter $\sigma_1^2+\sigma_2^2$.
    \item If $X_1, \dots X_m$ are i.i.d. mean zero sub-Gaussian random variables with parameter $\\sigma$, then,
    \begin{align*}
        P\p{\frac{1}{m}\Big|\sum_{i=1}^m X_i\Big|> \mu }\le 2\exp\p{-m \cdot \frac{\mu^2}{2\sigma^2}}
    \end{align*}
\end{enumerate}

\lemthresh*
\begin{proof} \ \\
    First, consider the case that $\omega= \omega_0$. For $\tau$ picked uniformly at random from $\Z_N$, let $X$ be the random variable given by 
    \begin{align*}
        X&= N(S\cdot e_{-\omega})[\tau].
    \end{align*}
    Then, $\mathbb{E}(X) = \alpha_0$. And moreover, since 
    \begin{align*}
        |\sqrt{N} \cdot \alpha_0e_{\omega_0}[\tau]|\le |\alpha_0|,
    \end{align*} the properties of independent sub-Gaussian random variables tells us that $X$ is \textit{sub-gaussian} with parameter $4|\alpha_0|^2+\sigma^2$. \\
    Let $X_1, \dots X_m$ correspond to $\tau_1, \dots, \tau_m$ chosen uniformly at random from $\Z_N$. Then, it follows from the properties of sub-Gaussian random variables that,
    \begin{align*}
        P\p{\frac{1}{m}|\sum_{i=1}^m X_i - \alpha_0| > \mu|\alpha_0|} &\le 2 \exp\p{-m\cdot\frac{\mu^2|\alpha_0|^2}{8|\alpha_0|^2+2\sigma^2}}\\
        & = 2 \exp\p{-m\cdot\frac{\mu^2}{8+2SNR^{-1}}}.
    \end{align*}
    But we have,
    \begin{align*}
        \widehat{\alpha} &= \frac{1}{m}\sum_{i=1}^m X_i,
    \end{align*}
    and since $||\widehat{\alpha}|- |\alpha_0||\ge |\widehat{\alpha} -\alpha_0|$ by the triangle inequality, the result follows.
\end{proof}

\subsection{Case 2: Discrete filter functions}
\lemfiltering*
\begin{proof} \ 
    \begin{enumerate}
        \item Since $m$ divides $N$ we have a natural map $\Z_N \to \Z_m$, namely for $\tau\in \Z_N$,
        \begin{align*}
            \tau \mapsto \tau \text{ mod } m.
        \end{align*}
        This induces a map which we will denote, $\text{Avg}_m:L^2(\Z_N)\to L^2(\Z_m)$, where,
        \begin{align*}
            (\text{Avg}_mS)[\tau] = \sum_{i = 1}^{N/m} S[\tau+i\cdot m]
        \end{align*}
        A quick check shows that for any $\tau \in \Z_N$,
        \begin{align*}
            S_j[\tau] & = \mathcal{F}\p{\text{Avg}_m\p{S\cdot H_{-\tau,0}F}}[j].
        \end{align*}
        Since $F$ has support of size $k$, computing $\text{Avg}_m\p{S\cdot H_{-\tau,0}F}$ requires at most $k$ samples of $S$ and $k$ arithmetic operations. Next, using FFT we can compute $S_1[\tau], \dots, S_m[\tau]$ with an additional $m\log m$ operations, and the result follows.
        \item Let $I\subset \Z_N$ denote the essential support of $F$. Define $\widetilde{F}$ to be the function,
        \begin{align*}
            \widetilde{F}[\tau] &= 
            \begin{cases}
                F[\tau] & \text{ if }\tau \in I \\
                0 &\text{ otherwise.}
            \end{cases}
        \end{align*}
        Then, 
        \begin{align*}
            \widetilde{S}_j = S * \p{\widetilde{F}\cdot e_{\frac{N}{m}j}}
        \end{align*} 
        is an estimate for $S_j$. By part 1, for any $\tau \in \Z_N$, we can compute $\widetilde{S}_1[\tau]$ using at most $k$ samples of $S$ and $k+m\log m$ arithmetic operations.
        \\
        Now we consider $|S_j[\tau]- \widetilde{S}_j[\tau]|$ for $j=1$. 
        \begin{align*}
            |S_1[\tau]- \widetilde{S}_1[\tau]| 
            & = \Big|\sum_{\tau'\in \Z_N} S[\tau -\tau']\cdot (F-\widetilde{F})[\tau']\Big| & \text{(definition of convolution)} \\
            & \le \p{\max_{\tau \in \Z_N} S[\tau]} \cdot \sum_{\tau'\notin I}\big|F[\tau']\big| &\text{(triangle inequality)} \\
            & \le 1 \cdot \sum_{\tau'\notin I}\big|F[\tau']\big| &\text{(since }S \text{ has unit norm)} \\
            & = O(1/N) &\text{(definition of essential support.)}
        \end{align*}
        The result follows by an identical argument for $j\ne 1$.  
    \end{enumerate}
    
\end{proof}

\lemfilnoise*
\begin{proof} \ \\
    Let $\nu$ and $\widetilde{\nu}$ denote the following random vectors,
    \begin{align*}
        \nu[\tau] = \nu_\tau, \ 
        \widetilde{\nu}[\tau] = \widetilde{\nu}_\tau.
    \end{align*}
    So we have,
    \begin{align*}
        (R*F) &= (S*F) + \widetilde{\nu}.
    \end{align*}
    The fourier transform then gives us,
    \begin{align*}
        \mathcal{F}(R*F) &= \mathcal{F}(S*F) + \mathcal{F}(\widetilde{\nu}),
    \end{align*}
    Let's consider $\mathcal{F}(\widetilde{\nu})$,
    \begin{align*}
        \mathcal{F}(\widetilde{\nu}) & = \mathcal{F}(\nu)\cdot \mathcal{F}(F) \\
        \mathbb{E}(\|\mathcal{F}(\widetilde{\nu})\|^2) &\approx \mathbb{E}(\sum_{\omega \in [-N/2k,N/2k]} |\mathcal{F}(\nu)\cdot \mathcal{F}(F)[\omega]|^2) & \text{(essential support)} \\
        & \le \mathbb{E}(\sum_{\omega \in [-N/2k,N/2k]} |\mathcal{F}(\nu)[\omega]|^2) & (|\mathcal{F}F[\omega]|\le 1)\\
        & = \sigma^2/k
    \end{align*}
    But we also have, 
    \begin{align*}
        \mathbb{E}(\|\mathcal{F}(\widetilde{\nu})\|^2)  & = \mathbb{E}(\|\widetilde{\nu}\|^2) \\
        &\ge N \cdot \widetilde{\sigma}^2/N &\text{(covariances are positive, so ignored)}\\
        & = \widetilde{\sigma}^2
    \end{align*}
    So altogether, we have that $\widetilde{\sigma}^2\le \sigma^2/k$.
\end{proof}

\subsection{General Case}
\lemisolate*
\begin{proof}
Given a pseudorandom permutation $\sigma$ picked uniformly at random,
\begin{align*}
\Pr (\omega_i, \omega_j \text{ collide after }\sigma) &\le \Pr (\sigma(\omega_i)-\sigma(\omega_j) \in [-C/2,C/2]) \\
& \le \frac{2C}{N}
\end{align*}
By a union bound, 
\begin{align*}
\Pr(\omega_i \text{ is not isolated by }\sigma) &\le \frac{2Ck}{N}
\end{align*}
\end{proof}

\corisolate*
\begin{proof}
By a union bound, 
\begin{align*}
\Pr(\omega_i \text{ is not isolated by }\sigma) &\le \frac{2Ck}{N} \le \frac{1}{4}  &(\text{by Lemma \ref{lem:isolate} and assumption on }C)
\end{align*}
Given $T =\log_4 (k/\delta) 
= O(\ln (k/\delta))$ permutations chosen independently uniformly at random,
\begin{align*}
\Pr(\omega_i \text{ is not isolated after }\sigma_t, \forall t = 1, \dots, T  ) &\le  \p{\frac{1}{4}}^{\log_4 (k/\delta)} = \frac{\delta}{k} 
\end{align*}
By a union bound,
\begin{align*}
\Pr(\exists i \text{ s.t. } \omega_i \text{ is not isolated after }\sigma_t, \forall t = 1, \dots, T  ) &\le \delta
\end{align*}
\end{proof}

Before we prove the guarantees of \cref{thm:SCE}, it may be in our interest to first consider the following schematic depiction for the SFFT process.
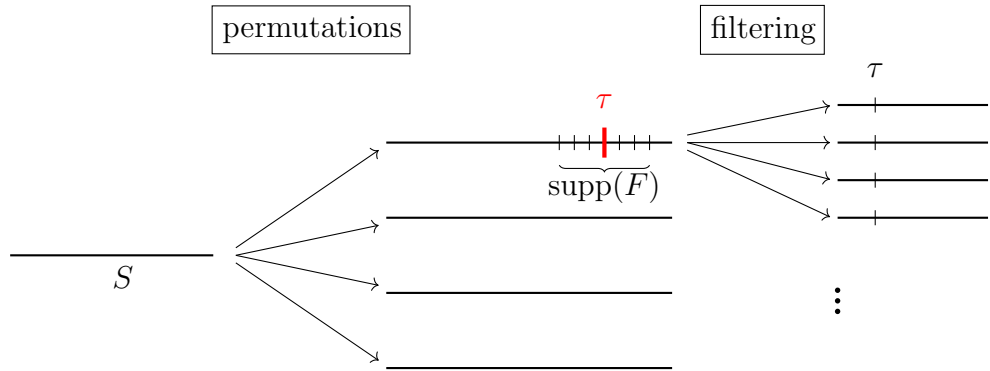
\begin{figure}[H]
\centering
\begin{tikzpicture}
	\draw[thick] (0,-1.5) -- (1.5,-1.5)node[below]{$S$} -- (2.7,-1.5);
	\draw[thick] (5,0) -- (8.8,0);
	
	\foreach \x in {7.3,7.5,7.7,8.1,8.3,8.5}
	\draw (\x,.1) -- (\x,-.1) node[below]{};
	\draw[red,ultra thick ] (7.9,.2) -- (7.9,-.2) node[above, yshift=1.2em]{$\tau$} ;
	\draw [decorate, decoration={brace,mirror}](7.3,-0.3) -- (8.5,-0.3) node[midway,yshift=-.7em]{$\text{supp}(F)$};
	
	\foreach \x in {-1,-2,-3}
	\draw[thick] (5,\x) -- (8.8,\x);
	
	\foreach \x/\y in {-1.4/-0.1,-1.5/-1.1,-1.5/-1.9,-1.6/-2.9} 
    	\draw[->](3,\x) -- (4.9,\y);

    	\foreach \x in {.5,0,-.5,-1}
    	\draw[thick](11,\x) -- (13,\x);

	\foreach \x/\y in {.1/.49, 0/0, 0/-.49, -.1/-.99}
    	\draw[->](9,\x) -- (10.9,\y);
    	
    	\foreach \x in {0,-.5,-1}
    	\draw (11.5,\x+.1) -- (11.5,\x-.1);
    	\draw (11.5,.5+.1) -- (11.5,.5-.1)node[above, yshift=.80em]{$\tau$};
    	
    	\node[draw,text width=67] at (4,1.5) {permutations};
    	\node[draw,text width=39] at (10,1.5) {filtering};
    	\node[] at (11,-2) {\Huge$\vdots$};
\end{tikzpicture}
\caption{A schematic representation of the SFFT process.} \label{fig:cplxtysfft}
\end{figure}

\begin{lemma}\label{lem:ugbutuseful} 
For a given implementation of $\text{SFFT}_\mu$ that utilizes a filter $F$, let 
\begin{enumerate}
    \item $n_{perm}$ denote the number of pseudorandom spectral permutations,
    \item $n_F$ denote the size of the support of $F$, 
    \item $n_{I}$ the number of intervals that the function is then filtered to and,
    \item $n_{1\text{SFFT}}$ denote the number of samples that the $1$-sparse algorithm requires.
\end{enumerate}
Then, $\text{SFFT}_\mu$ exhibits the following,
\begin{align*}
     \text{Sample complexity} &\le n_{perm}\cdot n_F\cdot n_{1\text{SFFT}}\\
     \text{Storage complexity} &\le c\cdot n_{perm}\cdot n_F\cdot n_{1\text{SFFT}} \\
     \text{Arithmetic complexity} &\le n_{perm}\cdot\p{n_F\cdot + n_I\log n_I}\cdot n_{1\text{SFFT}},
\end{align*}
where $c$ is a universal constant.
\end{lemma}
\begin{proof} \ \\
Our verification will utilize \cref{fig:cplxtysfft}. Let us first examine sample complexity:
\\The figure depicts that we start with a signal $S$, after a pseudorandom permutation we have a signal $S^{\sigma,a}$, say, and after filtering we will have $n_I$ signals $S_j$, $j=1, \dots, n_I$. We apply the $1\text{SFFT}_\mu$ algorithm to the $S_j$. We also see that, for any $\tau \in \Z_N$, we compute $S_j[\tau]$, $j=1, \dots, n_I$, with $n_F$ samples of $S^{\sigma,a}$. Therefore, after each permutation we use at most $n_F\cdot n_{1\text{SFFT}}$ samples and so, overall, at most $n_{perm}\cdot n_F\cdot n_{1\text{SFFT}}$ samples are used.

The arithmetic complexity of $\text{SFFT}_\mu$ can be reasoned in largely the same manner as above, the only difference being that the filtering step may require up to $\p{n_F\cdot + n_I\log n_I}$ arithmetic operations (\cref{lem:filtering}).

Storage complexity is at most a constant multiple of sample complexity, since nothing other than samples is stored.
\end{proof}

\vspace{0.7cm}
In order to provide a statistical guarantee and complexity bounds for Sparse Channel Estimation (SCE), we first provide analogous guarantees for $\text{SFFT}_\mu$. These guarantees will hold under certain assumptions, namely, we will assume the model \cref{eq:sfft} with the following additional features.
\begin{itemize}
\item We make a sparsity assumption on the number of significant frequencies.

\textbf{A21} (\textit{Sparsity}): The number of significant frequencies is at most $k$, where $k\ll N$ is a constant, i.e., independent of $N$.

\item We also make the following assumption on the coefficients $\alpha_j$ and their distribution.

\textbf{A22} (\textit{$\epsilon$-coefficients}):
There is some $A>0$ and $\epsilon\in (0,1)$ such that $(\alpha_1, \dots,\alpha_k)$ is drawn uniformly at random from the following set:
\begin{align*}
    B_{\epsilon} & = \set{x\in \C^k}:{\|x\|^2 = A \text{ and } \min_{x_j\ne 0} |x_j|\ge \epsilon\cdot \sqrt{A/k}}.
\end{align*}

\item The final assumption that we make is on the distribution of the noise $\nu_\tau$.

\textbf{A23} (\textit{sub-Gaussian}): We assume 
$\nu_\tau$ are i.i.d, mean zero and \textit{subgaussian} random variables with subgaussian parameter $\sigma^2/N$. 
\end{itemize}

For our purposes, we define \textit{signal-to-noise ratio}, or \textit{SNR} for short, to be 
\begin{align*}
    \text{SNR} = A/\sigma^2.
\end{align*} 
Moreover, we define the \textit{probability of detection} as follows.
\begin{definition} The \underline{probability of detection} (PD) of an SFFT algorithm is the probability that the $j^{th}$ frequency, $\omega_j$, is returned by the algorithm.
\end{definition}
\begin{theorem}[SFFT]\label{thm:SFFT}Let $\mu = \kappa\cdot \epsilon\sqrt{A/k}$ for some confidence parameter $\kappa\in (0,1)$. \\ Then, under the sparsity, $\epsilon$-coefficients, and subgaussian assumptions there is an implementation of $\text{SFFT}_\mu$ which takes
\begin{enumerate}
    \item $c_1k (\log N)^3(\epsilon^{2}\text{SNR} )^{-1}$ samples,
    \item $c_2 k (\log N)^3(\epsilon^{2}\text{SNR})^{-1}$ bits of memory, and
    \item  $c_3 k (\log N)^3(\epsilon^{2}\text{SNR})^{-1}+k^2$ arithmetic operations 
\end{enumerate}
for which $\textit{PD} \to 1$ as $N\to \infty$, where $c_1,c_2,c_3$ are constants independent of $\epsilon$, $\textit{SNR}$, $k$ and $N$. 
\end{theorem}
\begin{proof} \ \\
    We consider the event that the $j^{th}$ frequency $\omega_j$ in \cref{eq:sfft} is not detected. This may happen in the following scenarios:
    \begin{enumerate}
        \item $\omega_j$ was never isolated during the SFFT process.
        \item $\omega_j$ was isolated but incorrectly estimated by the SFFT process.
        \item $\omega_j$ was isolated and accurately estimated but the estimate for the $j^{th}$ coefficient $|\widehat{\alpha_j}|<\mu$.
    \end{enumerate}
    Therefore, we have the following estimate for $PD$,
    \begin{align*}
        PD &&\ge &&1-P(\omega_j\text{ not isolated}) +
                    P(\omega_j\text{ not located after isolation}) \\
                    &&&&+P( |\widehat{\alpha_j}|<\mu \text{ after isolation and location}).
    \end{align*}
    Using the notation of \cref{lem:ugbutuseful}, we would first like to understand how many samples we may require for $1\text{SFFT}_\mu$. Assume that $\omega_j$ has been isolated, then using a $(k,N/k,\delta)$-filter we will filter to $n_I$ intervals where,
    \begin{align*}
        n_I &= c_1 \cdot \frac{N}{2N\sqrt{\log(1/\delta)}/ k\sqrt{\log N} } \\
            &= c_1 \cdot k \cdot \sqrt{\log N/\log(1/\delta)}
    \end{align*}
    Then, by definition of such a filter and by \cref{lem:filnoise}, we can assume model \cref{model:k1noise} with signal-to-noise ratio,
    \begin{align*}
        SNR_1 \ge \frac{\epsilon^2\delta^2A}{\sigma^2}
    \end{align*}
    Let's simply pick $\delta = 1/\sqrt{2}$, then,
        \begin{align}
            n_I & = c_1 \cdot k\log N \label{eq:snr1}\\
            SNR_1 &\ge \frac{\epsilon^2A}{2\sigma^2} = \epsilon^2SNR^{-1}/2 \label{eq:snr2}
        \end{align}
    \begin{enumerate}
        \item for $n_{perm} = c_2\cdot \log(kN)$, by \cref{cor:isolate},
        $$P(\omega_j\text{ not isolated})\le 1/N,$$
        \item for $n_{1\text{SFFT}}\ge c_3 \cdot \log(N\log N)\cdot\epsilon^2SNR^{-1} $, by \cref{eq:snr2,lem:bitbybit},
         $$P(\omega_j\text{ not located after isolation})\le 1/N,$$
        \item for $n_{1\text{SFFT}}\ge \log N \cdot\epsilon^2SNR^{-1}/(1-\kappa)$, by \cref{eq:snr2,lem:thresh},
        $$P( |\widehat{\alpha_j}|<\mu \text{ after isolation and location})\le 1/N.$$
    \end{enumerate}
    Moreover, $n_F = k\log N$.\\
    Then, by \cref{lem:ugbutuseful}, this implementation of $\text{SFFT}_\mu$ takes at most,
    \begin{align*}
        &c_1\cdot \log(kN) \cdot k\log N \cdot \log(N\log N)\cdot \epsilon^2 SNR^{-1} \\
        &= c_1 k\log(N)^3\epsilon^2 SNR^{-1} 
    \end{align*}
    samples, and
    \begin{align*}
        &c_3 \cdot \log(kN) \cdot\p{k\log N+k\sqrt{\log N} \cdot \log\p{k\sqrt{\log N}}} \cdot \log(N\log N)\cdot \epsilon^2 SNR^{-1}\\
        & = c_3k\log(N)^3 \epsilon^2 SNR^{-1}
    \end{align*}
    arithmetic operations.
    \\
    Finally, storage complexity is at most a constant multiple of sample complexity since nothing other than samples is stored, and so the result follows.
\end{proof}
    
\vspace{0.7cm}
We are now ready to prove a statistical guarantee and complexity bounds for Sparse Channel Estimation (SCE).
\thmSCE*
\begin{proof} \ \\
Assume that $\text{SCE}_\mu$ was executed by transmitting $S= S_L+S_M+S_K$ for three distinct lines $L,M$, and $K$, and $R$ given by \cref{eq:digital} under the sparsity, $\epsilon$-targets, and sub-Gaussian assumptions A1, A2, and A3. 

Now, consider $\mathcal{A}(S_L,R)$ on the line $K$. The assumptions A1, A2, and A3 immediately imply the sparsity, $\epsilon$-coefficients, and sub-Gaussian assumptions A21, A22, and A23, for $\mathcal{A}(S_L,R)$ on $K$.
So, $\text{SCE}_\mu$ can be executed using the implementation of $\text{SFFT}_\mu$ assumed in \cref{thm:SFFT}. Therefore, samples, storage and arithmetic operations are as claimed.

We begin with the probability of detection $PD$ of $\text{SCE}_\mu$ (\cref{alg:SCE}). The shift $(\tau_j,\omega_j)$ is not detected if, in particular, the corresponding peak of $\mathcal{A}(S_L,R)$, is not detected. 
The above holds true for $\mathcal{A}(S_M,R)$ and $\mathcal{A}(S_K,R)$ as well. So, denoting the probability of detection of $\text{SFFT}_\mu$ as $PD_{\text{SFFT}_\mu}$, a union bound gives us that the probability of detection of $\text{SCE}_\mu$ is at most,
\begin{align*}
    PD \ge 1 - 3(1-PD_{\text{SFFT}_\mu})
\end{align*}
By \cref{thm:SFFT}, we then have that $PD\to 1$ as $N\to \infty$.

Now let's consider the probability of false alarm $PFA$ of $\text{SCE}_\mu$. A shift $(\widetilde{\tau},\widetilde{\omega})$ is falsely returned by the algorithm in the following scenarios. 
\begin{enumerate}
    \item The choice of lines $L,M,K$ produce a false triple intersection. By \cref{lem:genline}, this happens with probability less than $k^3/(N-2)$ or,
    \item Peaks corresponding to $(\widetilde{\tau},\widetilde{\omega})$ are incorrectly detected for $\mathcal{A}(S_L,R)$, $\mathcal{A}(S_M,R)$ and $\mathcal{A}(S_K,R)$. It follows by \cref{lem:thresh} that, with appropriately chosen constants $c_1,c_2$ and $c_3$, this holds with probability at most $1/N^3$.  
\end{enumerate}
So, the given implementation of $\text{SCE}_\mu$ also exhibits $PFA \to 0$ as $N\to \infty$.
\end{proof}

\genline*
\begin{proof}
The shifted lines $ p_i+L, \  p_i+M$, $i= 1, \dots, k$, intersect in at most $2 \cdot \begin{pmatrix} k \\ 2\end{pmatrix} = k(k-1)$ points $p\ne p_i$. 

There are at most $k^2(k-1)$ choices of $K$, such that $p_i+K$ passes through such a point $p$. 

The probability that this happens is at most $\frac{k^2(k-1)}{N-2}\le \frac{k^3}{N-2}$.

\end{proof}

%
%
\bibliographystyle{siam}
\bibliography{thesis}

\end{document}